\documentclass[a4paper,11pt]{article}
\usepackage[utf8x]{inputenc}
\usepackage{natbib}
\usepackage[english,french]{babel} 
\usepackage{ucs}
\usepackage{textcomp}
\usepackage{dsfont}
\usepackage[T1]{fontenc}
\usepackage{amsmath, amsfonts, amsthm, amssymb}
\usepackage{fullpage}
\usepackage{float} 
\usepackage{array}
\usepackage{hyperref} 
\usepackage{siunitx}
\usepackage{graphicx}
\usepackage{url}
\usepackage{xcolor}
\usepackage{subfig}

\newtheorem{theo}{Theorem}[section]
\newtheorem{defi}[theo]{Definition}
\newtheorem{prop}[theo]{Proposition}

\newtheorem{remark}[theo]{Remark}

\newtheorem{hypo}[theo]{Hypothesis}
\newcommand{\itemr}{\item[$\rightarrow$]}
\newcommand{\itemb}{\item[$\bullet$]}
\newcommand{\R}{\mathbb{R}}
\newcommand{\E}{\mathbb{E}}
\newcommand{\dx}{\textrm{d}}

\title{Crawling migration under chemical signalling: a stochastic particle model}
\author{Christ\`{e}le Etchegaray \footnote{Institut de Math\'{e}matiques de Toulouse, CNRS UMR 5219, Universit\'{e} Paul Sabatier, 118, route de Narbonne, F-31062 Toulouse Cedex 9. ORCID 0000-0001-6805-2120.}         \and
        Nicolas Meunier \footnote{ MAP5, UMR CNRS 8145, Universit\'{e} Paris Descartes, Sorbonne Paris Cit\'{e}, 45 rue des Saints P\`{e}res, Paris Cedex 6 75270, France.} 
}

\begin{document}
\selectlanguage{english}
\maketitle

 \section{Introduction}
 
 Cell migration is a fundamental process involved in physiological phenomena such as the immune response and morphogenesis \citep{Anon2012Cell-crawling-m,morphogenesis}, but also in pathological processes, such as the development of tumor metastasis \citep{Friedl2003Tumour-cell-inv}. 
These functions are effectively ensured because cells are active systems that adapt to their environment. Indeed, their internal organization relies on multiscale interactions between polymers and molecules based on out-of-equilibrium reactions \citep{Lauffenburger1996Cell-migration:}. \par 
We are interested in this paper in cells crawling on an adhesive surface. They spread and form extensions called protrusions, that are mechanically coupled to the substrate thanks to molecular adhesion complexes. This way, internal forces can be transmitted to the substrate, leading to motion. 
 \par
Cell protrusions in the case of crawling are divided in two types: the \emph{lamellipodia} are wide and flat while \emph{filopodia} are finger-like extensions. It has been stated that the fluctuations in the protrusive activity are important for the long-term cell behaviour \citep{Caballero2014Protrusion-fluc,Mattila2008Filopodia:-mole,Krause2014Steering-cell-m}. In particular, some trajectories are very efficient to explore a large territory, while others are more Brownian-like. Typical models of cell trajectories (with ou without external cues) consist in active brownian particle models \citep{Stokes1991,Schienbein1993,Romanczuk2012Active-Brownian} that can be fitted to experimental trajectories to provide quantifications of trajectories. However, these models consist typically in a Langevin equation on the particle's velocity that includes a positive (phenomenological) feedback loop, so that no specific feature of cell migration is described. In \citep{etchegaray:tel-01533458}, a first stochastic model of cell trajectories based on the protrusive activity was built and was proven able to capture the diversity in the observed trajectories.
\par 

It is known that the environment is able to guide cell migration either mechanically (ridigity and adhesiveness of the substrate, obstacles) or chemically if some molecular specie attracts or repulse the cell (so-called chemotaxis phenomenon). In both cases, the cell senses its outside using molecular receptors at the membrane and at protrusions tips, leading to an intracellular response \citep{heckman2013filopodia,Lidke2005Reaching-out-fo}. This ability is fundamental, since many cells need to attain some targets by following signals (antibodies, pathogens, etc). However, the cell's reaction to a signal is not clear, since its internal self-polarisation machinery may contradict its external awareness. \par 
In this work, we use the 2D particle model for cell trajectories developped in \citep{etchegaray:tel-01533458} and take into account a gradient in attractive chemical signal, that may vary in time.
We show that the resulting stochastic model is a well-posed non-homoegeneous markovian process, and provide cell trajectories in different settings.

	\section{Model construction}
In this part, we construct the 2D stochastic trajectory model developped in \citep{etchegaray:tel-01533458}, and show how to take into account a gradient in external signal. 	
	
We choose large space and time scales compared to the cell size and activity, so that it is natural to adopt an active particle approach. The cell is therefore a point (its center of mass) in a microscopic setting, submitted to the \textbf{force balance principle}. 

\subsection{Velocity model}
At each time $t$, write $\vec{V}_{t}$ for the cell velocity, with polar coordinates $(v_{t},\theta_{t})$. The force equilibrium principle states that the vectorial sum of all forces applied to the cell equals zero. In this case appear
\begin{itemize}
\itemr a passive force: the \textbf{friction force} exerted by the substrate on the cell due to motion, that writes $\vec{f} = -\gamma \vec{V}_t$, with $\gamma$ the global friction coefficient.
\itemr active forces related to the protrusion process. Filopodial protrusions can be considered as good readouts  \citep{Caballero2014Protrusion-fluc}, so that in the following, we consider only \textbf{filopodial forces}.
\end{itemize}

\begin{figure}
\centering
\includegraphics[scale=0.7]{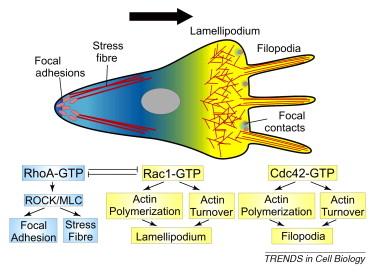}
\caption{Scheme of a polarised crawling cell: protrusive structures at the front (lamellipodium, filopodia), and contractile fibers at the back. The asymmetric activity combines with an asymmetric repartition of molecular regulators organized in feedback loops. Source: \citep{mayor2010keeping} }
\end{figure}

Finally, we obtain 

\begin{equation}\label{velocity_model}
\gamma \vec{V_t} = \sum_{i=1}^{N_t} \vec{F_i}(t)\, ,
\end{equation}
with $N_t$ the number of filopodia adhering on the substrate at time $t$, and $(\vec{F_i}(t))_{i}$ the filopodial forces. The cell motion is then entirely described by the protrusion forces.\par 

\begin{hypo}\label{hypo2}
Each filopodial force vector is unitary and constant in time.
\end{hypo}
Now, denoting $\theta_i = arg(\vec{F_i})$, one can write
\begin{displaymath}
\vec{F_i} = \begin{pmatrix}
\cos(\theta_i) \\ \sin(\theta_i)
\end{pmatrix}.
\end{displaymath}

Modelling cell motion then accounts to modelling the time evolution of the filopodial population, structured by their orientation.

\subsection{Protrusion model}
Each filopodium is characterized by a quantitative parameter, its orientation $\theta \in [0,2\pi)$. Therefore, we use a measure valued process for the stochastic evolution of the set of filopodia, as in ecological population models (\cite{fournier_microscopic_2004} and related works).
Let us denote $\mathcal{M}_F(\overline{\chi})$ the set of positive finite measures on $\overline{\chi}=[0,2\pi]$, equipped with the weak topology. Notice that as $\overline{\chi}$ is compact, weak and vague topologies on $\mathcal{M}_F(\overline{\chi})$ coincide. Write 
$\mathcal{M}$ for the subset of $\mathcal{M}_F(\overline{\chi})$ composed of all finite point measures. Then, a filopodium of orientation $\theta$ is described by a Dirac measure $\delta_\theta$ on $\chi$, and the whole population by 
\begin{displaymath}
\nu_t = \sum_{i=1}^{N_t} \delta_{\theta_i} \in \mathcal{M}.
\end{displaymath}

For any measurable function $f$ on $\overline{\chi}$ and any $\mu \in \mathcal{M}_F(\overline{\chi})$, we write $<\mu,f> = \int_{\chi} f(\theta) \mu(\dx \theta)$. In particular, $<\nu_t,f> = \sum_{i=1}^{N_t} f(\theta_i)$, and the population size corresponds to $N_t = <\nu_t,1>$.
A simple way to express the velocity equation \eqref{velocity_model} together with hypothesis \ref{hypo2} is to write

\begin{displaymath}
\gamma \vec{V_t} = \begin{pmatrix}
<\nu_t,\cos> \\ <\nu_t,\sin>
\end{pmatrix}.
\end{displaymath}

Now, the cell motion is entirely described by a measure-valued markovian jump process $(\nu_t)_t$, as in adaptive stuctured population models. We have now to define the events ruling its evolution. 

\begin{itemize}
\item First, filopodia \textbf{appear isotropically}. It is responsible for the spontaneous activity that is observed experimentally, that also probes the cell's surroundings. We write $\mathbf{c}$ for the appearance or creation rate.\par 
\item Each filopodium ends up disappearing: the disappearance or \textbf{death} rate is denoted by $\mathbf{d}$.
\item \textbf{Polarisation} is characterized by a morphological and functional asymmetry visible both on cell shape and at the microscopic scale. Here, we account for polarisation by its feedback on the protrusive activity. Two phenomena have to be distinguished:
\begin{itemize}
\itemr The formation of a protrusion is induced by several microscopic regulators and generates a \textbf{local} positive feedback on the protrusive machinery. Following that, we assume that each filopodium is able to \textbf{reproduce}. Denote $\mathbf{r(t,\theta_i,\nu_t)}$ the individual reproduction rate of a filopodium of orientation $\theta_i$ at time $t$.
 
\itemr \textbf{Polarisation} is also reinforced by intracellular actin flows, resulting from the protrusive activity and favoring the concentration of protrusions in a single direction, see \cite{Maiuri2015Actin_flows_med}. Moreover, it is stated in the same work that the space-averaged actin flow velocity over the cell is proportional to $\vec{V}$. As a consequence, we consider a reproduction rate that is positively coupled to $\alpha \vec{V}$ for $\alpha \geq 0$, imposing a non-local feedback. 
\end{itemize}
\item When a filopodium reproduces, the new protrusion may have the same orientation, or a slightly modified one due to the flucturations arising in the cell's internal signalling pathways involved. This is described by a \textbf{mutation event} for the orientation of the "offspring". We write $\boldsymbol{\mu}$ for the constant mutation rate. The new orientation is then chosen following a probability distribution $\mathbf{g(\cdot;\theta_i)}$ assumed centered in the parent's orientation $\theta_i$, with a constant variance.
\item Protrusions are cellular structures composed of actin polymers assembled in either networks or bundles, for which \textbf{the resources are limited}. Therefore, we assume that the more protrusions exist, the lower the creation and reproduction rates get: for a carrying capacity $\lambda$ of protrusions, we consider the creation rate $c\times (1-\frac{N_t}{\lambda})_+$, and individual reproduction rates of the form $r(t,\theta_i,\nu) \times (1-\frac{N_t}{\lambda})_+$. The formation of protrusions hence linearly decreases for an increasing protrusion population size. It is also clear that for an initial number of protrusions below $\lambda$, the population size will remain below it for all times, and the positive part is in this case not useful.
\end{itemize}

The possible events are summed up in the following graph:

\begin{center}
\centering
\begin{tabular}{cccclll}
Creation (global) &&&&  \\
\textcolor{blue}{$c\times (1-\frac{N_t}{\lambda})_+$} &&&&\\
   & & Clone &  &  \\
Reproduction (individual) &  $\nearrow$ & \textcolor{blue}{$1-\mu$} & & \\
 \textcolor{blue}{$r(t,\theta_i, \nu)\times (1-\frac{N_t}{\lambda})_+$} &  $\searrow$ &&&  \\
 && Mutation&  $\longrightarrow$ & Choice of $\theta$ \\
Death (individual)&& \textcolor{blue}{$\mu$} &&\textcolor{blue}{$ g(\theta ; \theta_i)$}\\
\textcolor{blue}{$d$}&&&&
\end{tabular}
\end{center}

\begin{remark}
When there is no interaction between individuals, the process $(N_t)_t$ simply follows an immigration, birth and death dynamics, and mathematical information can be derived. In particular, the branching property still holds. This is no longer the case when adding interactions. For example, if the interaction relies on $\vec{V}_t$, then knowing only $N_t$ is not sufficient and one has to know about the structured quantities $(N^{\theta_1},N^{\theta_2},...)$ at all time.
\end{remark}


\paragraph{Discussion}
The main ingredient of the model is the feedback of the cell velocity on the protrusive activity. We choose in this paper to include its effect in the local reproduction rate. Another choice consists in taking its effect into account rather in the global creation rate. From the modelling viewpoint, we believe that both models can be justified.

Let us discuss this fact in the absence of external signalling. 
If the feedback from the cell velocity affects the creation rate, and the reproduction rate is constant, for a large velocity new protrusions will appear more likely in the direction of motion, while "old" protrusions will continue reproducing, regardless of their orientation. Since the cell velocity induces a non-local effect, the model is relevant in the sense that it brings together same scale dynamics. \par 
In the model we study now, an isotropic creation dynamics is constantly at play for either slowing down or introducing new directions for the motion. This is more relevant in the spirit of the probing of the environment.

It is possible, using similar assumptions, to derive the same types of mathematical properties for both models. Their comparison will be the object of a future work.

\subsection{A reproduction rate for migration under chemical signalling}
In the present work, we want to investigate the effect of chemical signalling on the cell dynamics. The simplest framework consists in studying the effect of a homogeneous gradient of attractive signal in the medium. 

Chemical signals in the environment are detected by protrusions\citep{heckman2013filopodia}. The sensing corresponds to specific chemical reactions at the tip of filopodia, that propagates via molecular signalling loops towards the base of the protrusion. Therefore, the local dynamics reacts to the signal. In our modelling approach, we introduce a bias in the reproduction rate according to the direction of the gradient of signal.

\subsubsection{Reproduction rate without signal}
We have seen that in order to take into account cell polarisation, the individual reproduction rate must contain a positive correlation to $\alpha \vec{V}_t$. \par 
For cell migration without any external signalling, we used a reproduction rate proportional to the probability density of a circular normal distribution centered in $\theta_t$ the direction of motion, and getting sharper with an increasing velocity module $v_t$. More precisely, we had $r(\theta_i;\nu_t) = r^* f(\theta_i; \theta_t,\kappa(v_t))$ with 
\begin{displaymath}
f(\theta_i; \theta_t,\kappa(v_t)) = \frac{1}{2\pi I_0(\kappa(v_t))} \exp(\kappa(v_t) \cos(\theta_i - \theta_t))\, ,
\end{displaymath}
with $\kappa (v_t) = \beta \tanh(\alpha v_t) \geq 0$ quantifying the concentration in the direction of motion, and $I_0$ the $0$-order modified Bessel function of the first kind. The mutation law $g(\cdot;\theta_i)$ has also a circular normal distribution centered in $\theta_i$ and with a constant shape parameter (resp. variance) $K$ (resp. $\sigma^2$). \par 

In the present work, adding the effect of a homogeneous gradient of attractive signal amounts to introducing another direction favored by the reproduction rate. 

\subsubsection{Gradient of signal}
In this work, we consider that the cell is reacting to a constant gradient in concentration of chemical signals, that induces a constant bias in the protrusive activity. More precisely, we assume that the signal interferes with the direction favored by the reproduction rate. For that purpose, we consider a linear interpolation of two circular normal distributions, one being centered in the (variable) direction of motion, and the other in the direction of the gradient $\theta_g$. We write

\begin{equation}
r(t,\theta_i;\nu_t) = r^* \left( \left[1-\frac{\kappa_2(t)}{2\beta}\right] f(\theta_i; \theta_t,\kappa(v_t)) + \frac{\kappa_2(t)}{2\beta} f(\theta_i; \theta_g,\kappa_2) \right) \,,
\end{equation} 
with $\theta_g \in [0,2\pi)$ the direction of the gradient, and $\kappa_2$ a positive function of time, bounded from above by $2\beta$, quantifying the sensitivity to the chemical signal. Note that in this configuration, if there is no signal ($\kappa_2 \equiv 0$), then the reproduction rate is equal to the classical choice in the trajectory model. If the sensitivity to the signal rises, there is a balance between the cell's persistence and the chemical sensing. \par

The following figure \ref{fig:ReproductionRate} shows the reproduction rate's function for $\theta_t=\frac{\pi}{2}$, and $\theta_g=0$. We take $\kappa(v_t) \equiv k_1$ given below and choose values of $\kappa_2$ around it. We can see the change of balance from the internal polarisation to the signal sensing, in both the low and high polarised case. 

\begin{figure}[H]
\centering
\includegraphics[scale=0.3]{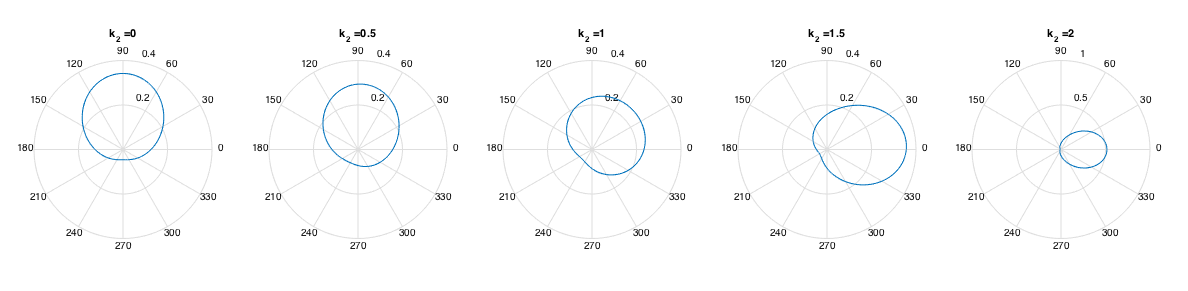}\\
\includegraphics[scale=0.3]{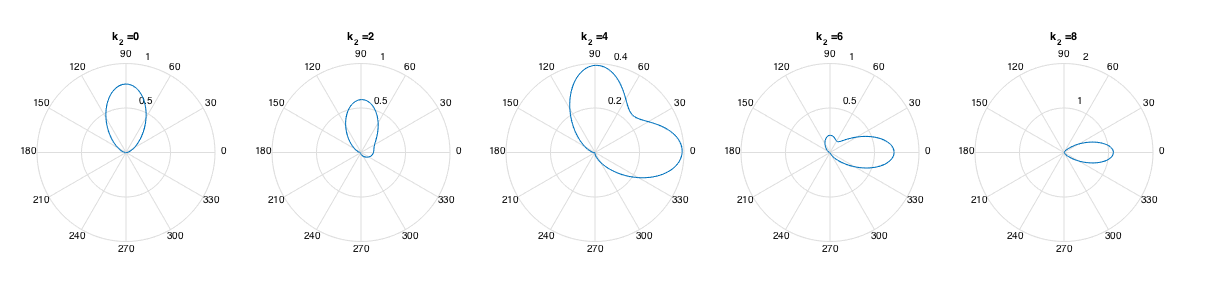}
\caption{Up : $\kappa_1 = 1$. Down : $\kappa_1 = 4$. }\label{fig:ReproductionRate}
\end{figure}

\section{Mathematical properties}
From now on, we will use the notation $C$ for any constant, that will change from line to line. We introduce here a stochastic differential equation for $(\nu_t)_t$ driven by Point Poisson Measures. We will show existence and uniqueness of a solution, and prove that it follows the dynamics previously described. 

In order to pick a specific individual in the population, we have to be able to order them or their trait. Indeed, from the n-uplet $(\theta_1,...,\theta_N)$, one can recover $\nu = \sum_{i=1}^N \delta_{\theta_i}$, but from $\nu$ it is only possible to know $\{\theta_1,...,\theta_N\}$. We do so following the notation used in \cite{fournier_microscopic_2004}.

\begin{defi}\label{notation}
Let us define the function 
\begin{displaymath}
\begin{array}{lccl}
H=(H^{1},...,H^{k},...) : & \mathcal{M} & \longrightarrow & \left(\chi\right)^{\mathbb{N}^{*}} \\
&&&\\
    &\displaystyle \nu = \sum_{i=1}^{n} \delta_{\theta_{i}} & \longmapsto & (\theta_{\sigma(1)},...,\theta_{\sigma(n)},... ), 
\end{array}
\end{displaymath}
with $\theta_{\sigma(1)} \preceq \theta_{\sigma(2)}\preceq ... \preceq \theta_{\sigma(n)}$, for an arbitrary order $\preceq$. 
Now, an individual can be picked by its label $i$, and the corresponding trait writes $H^{i}(\nu) = \theta_{\sigma(i)}$.
\end{defi}

Let $(\Omega,\mathcal{F},\mathbb{P})$ be a probability space, and $n(\dx i)$ the counting measure on $\mathbb{N}^*$. We introduce the following objects: 

\begin{itemize}\label{objets_proba}
\itemb $\nu_0\in \mathcal{M}$ the finite point measure describing the initial population, eventually equal to the null measure. It can be chosen stochastic as soon as $\mathbb{E}[<\nu_0,1>]<+\infty$.
\itemb $M_0(\dx s,\dx\theta,\dx u)$ a Poisson Point Measure on $[0,+\infty) \times \chi \times \mathbb{R}_+$, of intensity measure $\dx s \,\dx\theta\, \dx u$,
\itemb $M_1(\dx s,\dx i,\dx u)$ and $M_3(\dx s,\dx i,\dx u)$ Poisson Point Measures on $[0,+\infty) \times \mathbb{N}^* \times \mathbb{R}_+$, both of intensity measure $\dx s\, n(\dx i) \,\dx u$,
\itemb $M_2(\dx s,\dx i,\dx\theta,\dx u)$ a Poisson Point Measure on $[0,+\infty) \times \mathbb{N}^* \times \chi \times \mathbb{R}_+$, of intensity measure $\dx s n(\dx i) \dx\theta \dx u$.
\end{itemize}
The Poisson Measures are independent. Finally, $(\mathcal{F}_t)_{t\geq 0}$ denotes the canonical filtration generated by these objects. Let us construct the  $(\mathcal{F}_t)_{t\geq 0}$-adapted process $(\nu_t)_{t\geq 0}$ as the solution of the following SDE: $\forall t\geq 0,$ and writing $N_s = <\nu_s,1>$, we have

\begin{equation}\label{processus}
\begin{array}{lllllll}
\nu_t &=& \nu_0 & &&& \\
 &+&\displaystyle  \int_{0}^t \int_{\chi \times \mathbb{R}_+} &\delta_{\theta} && \displaystyle \mathds{1}_{u \leq c\times (1-\frac{N_s}{\lambda})_+} & M_0(\dx s,\dx\theta,\dx u) \\
&+& \displaystyle \int_{0}^t \int_{\mathbb{N}^* \times \mathbb{R}_+} & \delta_{H^i(\nu_{s})} &\displaystyle \mathds{1}_{i \leq N_{s}} & \mathds{1}_{u \leq (1- \mu)r(s,H^i(\nu_{s}),\nu_s)(1-\frac{N_s}{\lambda})_+ } &  M_1(\dx s,\dx i,\dx u) \\  
&+& \displaystyle \int_{0}^t \int_{\mathbb{N}^* \times \chi \times \mathbb{R}_+} &  \delta_{\theta} & \displaystyle \mathds{1}_{i \leq N_{s}} &\displaystyle \mathds{1}_{u \leq \mu r(s,H^i(\nu_{s}),\nu_s) (1-\frac{N_s}{\lambda})_+ g\left( \theta; H^i(\nu_{s})\right))} & M_2(\dx s,\dx i,\dx\theta,\dx u) \\
&-&\displaystyle \int_{0}^t \int_{\mathbb{N}^* \times \mathbb{R}_+} &\delta_{H^i(\nu_{s})} & \displaystyle \mathds{1}_{i \leq N_{s}} & \displaystyle \mathds{1}_{u \leq d} & M_3(\dx s,\dx i,\dx u).
\end{array}
\end{equation}

In this equation, each term describes a different event. The Poisson Point Measures generate atoms homogeneously in time. However, the dynamics we want to describe follows state and time-dependent rates. Hence, we use indicator functions to keep only some of the events in order to get the wanted rates. Then, the Dirac measures correspond to the individuals added to or removed from the population. \par 

In this work, we assume that the reproduction rate is a bounded function: $\exists \overline{r}>0$ such that $\forall \nu \in \mathcal{M}$, $\forall (\theta,\nu) \in \chi \times \mathcal{M}$, $\forall t\in \R_+$, $0\leq r(t,\theta,\nu) \leq \overline{r}$.

\subsection{Existence and uniqueness}
Let us now show existence and uniqueness of a solution for equation (\ref{processus}).
\begin{prop}\label{existence}
Recall that $N_t = <\nu_t,1>$. Assume the boundedness of rates, and that $\mathbb{E}[N_0]<+\infty$.
\begin{enumerate}
\item There exists a solution $\nu \in \mathbb{D}(\mathbb{R}_+,\mathcal{M}(\chi))$ of equation (\ref{processus}) such that
\begin{equation}\label{controlNt}
\forall T>0,\; \mathbb{E}\left[\sup_{t\in [0,T]} N_t\right]<\mathbb{E}\left[ N_0\right]e^{\overline{r}T} + \frac{c}{\overline{r}}(e^{\overline{r}T}-1)<+\infty\,. 
\end{equation}
\item There is strong (pathwise) uniqueness of the solution. 
\end{enumerate}
\end{prop}

\begin{proof}[Proof of \ref{existence}]
The proof is similar to prop.2.2.5 and 2.2.6 in \cite{fournier_microscopic_2004}.

\begin{enumerate}
\item Let $T_0=0$, and $t\in \mathbb{R}_+$. Then, the global jump rate of $\nu_t$ is smaller than $c + (\overline{r}+d)N_t$. Hence one can $\mathbb{P}-a.s$ define the sequence $(T_k)_{k\in\mathbb{N}^*}$ of jumping times, as well as 
$T_\infty := \lim_{k\rightarrow +\infty} T_k$.\par 

Now, by construction, it is $\mathbb{P}-a.s$ possible to build "step-by-step" a solution of equation (\ref{processus}) on $[0,T_\infty[$. Showing existence of a solution $(\nu_t)_{t\in \mathbb{R}_+} \in \mathbb{D}(\mathbb{R}_+,\mathcal{M}(\chi))$ amounts to showing that $\mathbb{P}-a.s$, $T_\infty = +\infty$. That is equivalent to saying that there cannot be an infinite number of jumps in a finite time interval. \par 

\itemb \textbf{\underline{First, we show the control property (\ref{controlNt})}}. For $n>0$ define the sequence of stopping times $(\tau_n)_n$ by

\begin{displaymath}
\tau_n = \inf_{t \geq 0} \{N_t \geq n\}.
\end{displaymath}

\begin{itemize}
\item[$\rightarrow$] \boldmath \textbf{Let us show that $(\tau_n)_{n\geq 0} $  is a sequence of stopping times for $(\mathcal{F}_t)_t$}. \unboldmath Denote $\sigma_t=\sigma(\nu_s, \; 0\leq s \leq t)$ the $\sigma$-algebra generated by $\{\nu_s,\;0\leq s \leq t\}$. Then $\forall t\geq0$, $\sigma_t \subseteq \mathcal{F}_t$. For $(n,m)\in \left(\mathbb{N}^*\right)^2$, notice that
 
\begin{displaymath}
\begin{aligned}
\{\tau_n \leq m \} &= \{\inf\{t\geq 0 ,\; \left< \nu_t,1\right> \geq n \} \leq m \} \\
& \in \sigma_m \subseteq \mathcal{F}_m,
\end{aligned}
\end{displaymath}
and $(\tau_n)_{n\geq 0} $ is indeed a sequence of stopping times.

\item[$\rightarrow$]\boldmath \textbf{Now, we prove that for all $T<+\infty$, the quantity 
$\mathbb{E}\left[ \sup_{t\in [0,T \wedge \tau_n]} N_t \right]$ is bounded $\forall n\geq 0$.} \unboldmath \par

For $t\in\mathbb{R}_+$, using equation (\ref{processus}) and dropping the non-positive term, one has

\begin{eqnarray*}
N_{t\wedge \tau_n} &=& <\nu_{t\wedge \tau_n},1> \leq  N_0 + \int_{0}^{t\wedge \tau_n} \int_{\chi \times \mathbb{R}_+} \mathds{1}_{u \leq c\times (1-\frac{N_s}{\lambda})_+}  M_0(\dx s,\dx \theta,\dx u) \\
&+& \int_{0}^{t\wedge \tau_n} \int_{\mathbb{N}^* \times \mathbb{R}_+}  \mathds{1}_{i \leq N_{s}} \mathds{1}_{u \leq (1- \mu)r(s,H^i(\nu_{s}),\nu_s)(1-\frac{N_s}{\lambda})_+} M_1(\dx s,\dx i,\dx u) \\  
&+& \int_{0}^{t\wedge \tau_n}\int_{\mathbb{N}^* \times \chi \times \mathbb{R}_+}  \mathds{1}_{i \leq N_{s}}  \mathds{1}_{u \leq \mu r(s,H^i(\nu_{s}),\nu_s)(1-\frac{N_s}{\lambda})_+  g( \theta; H^i(\nu_{s}))} M_2(\dx s,\dx i,\dx \theta,\dx u).
\end{eqnarray*}

As each integrand is positive, bounded, and integrable with respect to the intensity measure, taking the expectation and using the Fubini theorem, we can write

\begin{eqnarray}\label{inequality1}
\mathbb{E}\left[\sup_{t\in [0,T\wedge \tau_N]}N_t\right] &\leq & \mathbb{E}[N_0] + \E\left[\int_{0}^{T\wedge \tau_N} \left( c + \sum_{i=1}^{N_t} r(t,\theta_i,\nu_t) \right)\dx t \right]  \nonumber\\
&\leq & \mathbb{E}[N_0] + cT + \overline{r} \int_{0}^{T} \E\left[\sup_{s\in [0,t\wedge \tau_N]}N_s\right] \dx t \nonumber\\
\end{eqnarray}
leading to the $T$-dependent bound using the Gronwall inequality. 

\item[$\rightarrow$] \boldmath \textbf{ Let us prove that $\mathbb{P}-a.s,\; \lim_{n\rightarrow +\infty} \tau_n = +\infty$.} \unboldmath If this wasn't the case, there would exist $M<+\infty$ and a set $A_M \subset \Omega$ such that $\mathbb{P}(A_M) >0$, and $\forall \omega \in A_M, \; \lim_{n\rightarrow +\infty} \tau_n(\omega) < M$. By the Markov inequality, $\forall T > M,$

\begin{displaymath}
\E\left[\sup_{t\in [0,T\wedge \tau_n]}N_t\right] \geq n \underbrace{\mathbb{P}\left(\sup_{t\in [0,T\wedge \tau_n]}N_t \geq n\right)}_{\geq \mathbb{P}(A_M) >0},
\end{displaymath}
which is in contradiction with equation \eqref{inequality1}. 

\item[$\rightarrow$] \textbf{Property (\ref{controlNt}) is proved by the Fatou lemma:} 
\begin{equation*}
\begin{aligned}
\E\left[ \sup_{t\in [0,T]}N_t \right] &= \E\left[\liminf_{n\rightarrow +\infty} \sup_{t\in [0,T\wedge \tau_n]}N_t \right] \\
&\leq  \liminf_{n\rightarrow +\infty} \E\left[ \sup_{t\in [0,T\wedge \tau_n]}N_t \right] \leq \mathbb{E}[N_0] e^{\overline{r}T} + \frac{c}{\overline{r}}\left( e^{\overline{r}T} -1\right) <+\infty.
\end{aligned}
\end{equation*}

\itemb \textbf{\underline{\boldmath Now, let us show that $\mathbb{P}-a.s$, $T_\infty = +\infty$. \unboldmath}} If this is not the case, then there exists $\overline{M} <+\infty$ and a set $A_{\overline{M}} \subset \Omega$ such that $\mathbb{P}(A_{\overline{M}}) >0$ and $\forall w \in A_{\overline{M}}$, $T_\infty(\omega) < \overline{M}$. Moreover, if the assertion

\begin{equation}\label{assertion1}
\forall \omega \in A_{\overline{M}},\; \lim_{k\rightarrow +\infty} N_{T_k}(\omega) = +\infty,
\end{equation}

is true, then we would have 

\begin{displaymath}
\forall N>0, \; \forall \omega \in A_{\overline{M}},\; \tau_N(\omega) \leq \overline{M},
\end{displaymath}

which contradicts $\lim_{n\rightarrow +\infty} \tau_n = +\infty$. As a consequence, if we prove (\ref{assertion1}), the proposition is proved. 
If (\ref{assertion1}) is not true, there would exist $N'>0$ and a set $B\subset A_{\overline{M}}$ such that $\mathbb{P}(B)>0$ and 

\begin{displaymath}
\forall \omega \in B, \; \forall k \in \mathbb{N},\; N_{T_k}(\omega) < N'.
\end{displaymath}

Then, $\forall \omega \in B$, $(T_k(\omega))_k$ can be seen as the subsequence of a sequence of jumping times $(T_k^1(\omega))_k$ of a Point Poisson Process of intensity $c + (\overline{r}+d)N'$. The only accumulation point of $(T_k^1(\omega))_k$ being $\mathbb{P}-a.s \; +\infty$, it contradicts the definition of $B$, and proves (\ref{assertion1}). 
\end{itemize}

\item The sequence of jumping times $(T_k)_{k\in \mathbb{N}}$ being already defined, we only have to show that $(T_k,\nu_{T_k})_{k\in \mathbb{N}}$ are uniquely determined by $D=(\nu_0,M_0,M_1,M_2,M_3)$ defined above.
But this is clear by construction of the process. 
\end{enumerate}
\end{proof}

\subsection{Markov property}

Now, we can show that the solution $(\nu_t)_t$ of equation (\ref{processus}) is a Markov process  in the Skorohod space $\mathbb{D}(\R_+,\mathcal{M}_F(\chi))$ of càdlàg finite measure-valued processes on $\chi$. For that purpose, we introduce $\forall \nu\in \mathcal{M}$, $\Phi: \mathcal{M} \rightarrow \mathbb{R}$ measurable and bounded, and $t\geq 0$, the operator $L_t$ defined by 

\begin{eqnarray}\label{generator}
L_t\Phi(\nu) &=& \int_{\chi} c\times (1-\frac{N_t}{\lambda})_+ \left[ \Phi(\nu + \delta_{\theta}) - \Phi(\nu) \right] \dx \theta \nonumber\\
&+& \int_{\chi} (1- \mu) r(t,\theta, \nu)(1-\frac{N_t}{\lambda})_+  \left[ \Phi(\nu + \delta_{\theta}) - \Phi(\nu) \right] \nu(\dx \theta) \\
&+& \int_{\chi} \mu r(t,\theta,\nu)(1-\frac{N_t}{\lambda})_+  \int_{\chi} \left[ \Phi(\nu + \delta_{z}) - \Phi(\nu) \right] g(z; \theta) \dx z \; \nu(\dx \theta)\nonumber\\
&+& \int_{\chi} d \left[\Phi(\nu - \delta_{\theta}) - \Phi(\nu) \right] \nu(\dx \theta). \nonumber 
\end{eqnarray}

\begin{prop}\label{markovprop}
 Take $(\nu_t)_{t\geq 0}$ the solution of equation (\ref{processus}) with $\mathbb{E}[<\nu_0,1>]<+\infty$. Then, $(\nu_t)_{t\geq 0}$ is a Markovian process of infinitesimal generator $L$.
\end{prop}
In particular, this proposition ensures that the law of $(\nu_t)_{t\geq 0}$ is independent of the order $\preceq$ involved in (\ref{notation}).

\begin{proof}[Proof of proposition \ref{markovprop}]
The process $(\nu_t)_{t\geq 0} \in \mathbb{D}(\mathbb{R}_+,\mathcal{M}(\overline{\chi}))$ is markovian by construction. Now, let $N_0 <N< +\infty$, and consider again the stopping time $\tau_N$. Let $\Phi :  \mathcal{M} \rightarrow \mathbb{R}$ be measurable and bounded. For simplicity, we express the infinitesimal generator at time $t=0$. As $\mathbb{P}-a.s$ we can write

\begin{equation}\label{Phi(nu)}
\Phi(\nu_t) = \Phi(\nu_0)  + \sum_{s\leq t} \Phi(\nu_{s^-} + (\nu_{s}-\nu_{s^-}))-\Phi(\nu_{s^-})\,, 
\end{equation}
we have

\begin{eqnarray*}
&& \Phi(\nu_{t\wedge \tau_N}) = \Phi(\nu_{0}) + \int_0^{t\wedge \tau_N} \int_{\chi\times \mathbb{R}_+} \left[ \Phi(\nu_{s^-} + \delta_{\theta}) -\Phi(\nu_{s^-} )  \right] \mathds{1}_{u\leq c(1-\frac{N_{s^-}}{\lambda})_+} M_0(\dx s,\dx \theta,\dx u) \\
&+& \int_0^{t\wedge \tau_N} \int_{\mathbb{N^*}\times \mathbb{R}_+} \left[ \Phi(\nu_{s^-} + \delta_{H^i(\nu_{s^-})}) -\Phi(\nu_{s^-})   \right] \mathds{1}_{i\leq N_{s^-}} \mathds{1}_{u\leq (1-\mu) r(s^-,H^i(\nu_{s^-}),\nu_{s^-})(1-\frac{N_{s^-}}{\lambda})_+} M_1(\dx s,\dx i,\dx u)  \\
&+& \int_0^{t\wedge \tau_N} \int_{\mathbb{N^*}\times  \chi \times \mathbb{R}_+} \left[ \Phi(\nu_{s^-} + \delta_{z}) -\Phi(\nu_{s^-})   \right] \mathds{1}_{i\leq N_{s^-}} \mathds{1}_{u\leq \mu r(s^-,H^i(\nu_{s^-}),\nu_{s^-})(1-\frac{N_{s^-}}{\lambda})_+ g( z; H^i(\nu_{s}))} M_2(\dx s,\dx i,\dx z, \dx u)  \\
&+& \int_0^{t\wedge \tau_N} \int_{\mathbb{N^*} \times \mathbb{R}_+} \left[ \Phi(\nu_{s^-} - \delta_{H^i(\nu_{s^-})}) -\Phi(\nu_{s^-})   \right] \mathds{1}_{i\leq N_{s^-}} \mathds{1}_{u\leq d} M_3(\dx s,\dx i, \dx u)\,.  \\
\end{eqnarray*}

Again, as all integrands are bounded, we can take expectations to get

\begin{eqnarray*}
\mathbb{E}\left[ \Phi(\nu_{t\wedge \tau_N}) \right] &=& \mathbb{E}\left[ \Phi(\nu_{0})  \right] + \mathbb{E}\left[  \int_0^{t\wedge \tau_N} \int_{\chi} \left[ \Phi(\nu_{s^-} + \delta_{\theta}) -\Phi(\nu_{s^-} )  \right] c\times (1-\frac{N_{s^-}}{\lambda})_+ \dx \theta \dx s \right] \\
&+& \mathbb{E}\left[  \int_0^{t\wedge \tau_N} \sum_{i=1}^{N_{s^-}}
 \left[ \Phi(\nu_{s^-} + \delta_{H^i(\nu_{s^-})}) -\Phi(\nu_{s^-})   \right]  (1-\mu) r(s^-,H^i(\nu_{s^-}),\nu_{s^-})(1-\frac{N_{s^-}}{\lambda})_+ \dx s  \right]\\
&+& \mathbb{E}\left[  \int_0^{t\wedge \tau_N} \sum_{i=1}^{N_{s^-}}
\mu r(s^-,H^i(\nu_{s^-}),\nu_{s^-})(1-\frac{N_{s^-}}{\lambda})_+ \int_{\chi } \left[ \Phi(\nu_{s^-} + \delta_{z}) -\Phi(\nu_{s^-})   \right]  g(z; H^i(\nu_{s}))  \dx z \dx s  \right]\\
&+& \mathbb{E}\left[  \int_0^{t\wedge \tau_N}  \sum_{i=1}^{N_{s^-}} \left[ \Phi(\nu_{s^-} - \delta_{H^i(\nu_{s^-})}) -\Phi(\nu_{s^-})   \right] 
d \, \dx s \right], \\
&=:& \mathbb{E}\left[ \Phi(\nu_{0}) \right] + \mathbb{E}\left[ \psi(t\wedge \tau_N,\nu) \right].
\end{eqnarray*}
On the one hand, $\forall t\in [0,T]$, 
\begin{eqnarray*}
\parallel \psi(t\wedge \tau_N,\nu) \parallel_{\infty} &\leq & 2 T \parallel \Phi \parallel_{\infty} c  + 2 T \parallel \Phi \parallel_{\infty} (1- \mu) \overline{r} N + 2 T \parallel \Phi \parallel_{\infty} \mu \overline{r} N + 2 T \parallel \Phi \parallel_{\infty} d N \\
&\leq & C T \parallel \Phi \parallel_{\infty} (c + (\overline{r}+d)N) < +\infty.
\end{eqnarray*}

On the other hand, $t\mapsto \psi(t\wedge \tau_N,\nu)$ is derivable in $t=0$ $\mathbb{P}-a.s$ (as $\nu \in \mathbb{D}(\mathbb{R}_+,\mathcal{M}(\chi))$), and for a given $\nu_0$, we have

\begin{eqnarray*}
\frac{\partial \psi}{\partial t}(0,\nu_0) &=&  \int_{\chi} \left[ \Phi(\nu_{0} + \delta_{\theta}) -\Phi(\nu_{0} )  \right] c\times (1-\frac{N_{0}}{\lambda})_+ \dx \theta \\
&+& \sum_{i=1}^{N_{0}}
 \left[ \Phi(\nu_{0} + \delta_{H^i(\nu_{0})}) -\Phi(\nu_{0})   \right]  (1-\mu) r(0,H^i(\nu_{0}),\nu_{0})(1-\frac{N_{0}}{\lambda})_+   \\
&+&  \sum_{i=1}^{N_{0}} \mu r(0,H^i(\nu_{0}),\nu_{0})(1-\frac{N_{0}}{\lambda})_+
\int_{\chi } \left[ \Phi(\nu_{0} + \delta_{z}) -\Phi(\nu_{0})   \right]  g( z; H^i(\nu_{0}))  \dx z  \\
&+&  \sum_{i=1}^{N_{0}} \left[ \Phi(\nu_{0} - \delta_{H^i(\nu_{0})}) -\Phi(\nu_{0})   \right] d\,.
\end{eqnarray*}

Moreover, $\parallel \frac{\partial \psi}{\partial t}(0,\nu_0) \parallel \leq C \parallel \Phi \parallel_{\infty}(c+ N_0 (\overline{r}+d))$. Now,

\begin{displaymath}
\begin{aligned}
L_0\phi(\nu_0) &:= \left. \frac{\partial \mathbb{E}\left[\phi(\nu_t)\right]}{\partial t}\right\vert_{t=0} \\
&= \int_{\chi} \left[ \Phi(\nu_{0} + \delta_{\theta}) -\Phi(\nu_{0} )  \right] c\times (1-\frac{N_{0}}{\lambda})_+ \dx \theta \\
&+ \sum_{i=1}^{N_{0}}
 \left[ \Phi(\nu_{0} + \delta_{H^i(\nu_{0})}) -\Phi(\nu_{0})   \right]  (1-\mu) r(0,H^i(\nu_{0}),\nu_{0})(1-\frac{N_{0}}{\lambda})_+  \\
 & +  \sum_{i=1}^{N_{0}} \mu r(0,H^i(\nu_{0}),\nu_{0})(1-\frac{N_{0}}{\lambda})_+
\int_{\chi } \left[ \Phi(\nu_{0} + \delta_{z}) -\Phi(\nu_{0})   \right] g( z; H^i(\nu_{0}))  \dx z
\\
& +  \sum_{i=1}^{N_{0}} \left[ \Phi(\nu_{0} - \delta_{H^i(\nu_{0})}) -\Phi(\nu_{0})   \right] d,
\end{aligned}
\end{displaymath}
or equivalently 

\begin{displaymath}
\begin{aligned}
L_0\phi(\nu_0) &= \int_{\chi} \left[ \Phi(\nu_{0} + \delta_{\theta}) -\Phi(\nu_{0})  \right] c\times (1-\frac{N_{0}}{\lambda})_+ \dx \theta \\
&+ \sum_{i=1}^{N_{0}}
 \left[ \Phi(\nu_{0} + \delta_{H^i(\nu_{0})}) -\Phi(\nu_{0})   \right]  r(0,H^i(\nu_{0}),\nu_{0}) (1-\frac{N_{0}}{\lambda})_+ \\
 & +  \sum_{i=1}^{N_{0}} \left[ \Phi(\nu_{0} - \delta_{H^i(\nu_{0})}) -\Phi(\nu_{0})  \right] d\\
 &+ \mu \sum_{i=1}^{N_{0}} r(0,H^i(\nu_{0}),\nu_{0})(1-\frac{N_{0}}{\lambda})_+ \left( \int_{\chi } \Phi(\nu_{0} + \delta_{z}) g( z; H^i(\nu_{0}))  \dx z - \Phi(\nu_{0} + \delta_{H^i(\nu_{0})}) \right) \,.
 \end{aligned}
\end{displaymath}
\end{proof}

\section{Numerical simulations}
\subsection{Simulation method}
Since the process $(\nu_t)_t$ is markovian, it is classical to simulate the occurrence of each event one after the other. We proceed as follows: start with the population measure $\nu_k$ at time $t_k$, for a particle located at $X_{k}$, and a protrusion population size $N_k=<\nu_k,1>$.
\begin{description}
\item[Time of next event.] Since the reproduction rate depends on time, the process $(\nu_t)_t$ is non-homogeneous: the global jumping rate at time $t_k$ writes $$\tau(t_k) =  (c + <\nu_k,r(t_k,\cdot,\nu_k)>)(1-\frac{N_k}{\lambda})_+ + dN_k\,.$$
 In order to get the time of the next event, we use the \textbf{thinning} method \citep{lewis1979simulation}. The idea is the following: we find the time of next event $t_{k+1}$ for a dominating Poisson process of rate $\overline{\tau} = (c + \overline{r}N_k)(1-\frac{N_k}{\lambda})_+ + dN_k$. We have $t_{k+1} = t_k + \Delta t$ with 
\begin{displaymath}
\Delta t \sim Exp(\overline{\tau})\,.
\end{displaymath}
\item[Nature of the event.] What happens at time $t_{k+1}$ is determined as follows:
\begin{itemize}
\item creation of a protrusion occurs with probability $\frac{c\times (1-\frac{N_k}{\lambda})_+}{\overline{\tau}}$. Its orientation is chosen uniformly on $[0,2\pi)$. 
\item reproduction of the protrusion number $i$ occurs with probability $\frac{r(t_{k+1},H^i(\nu_k),\nu_k)(1-\frac{N_k}{\lambda})_+}{\overline{\tau}}$. This is where the time dependance is taken into account. Then,
\begin{itemize}
\item[$\rightarrow$] with probability $(1-\mu)$, the new protrusion has orientation $H^i(\nu_k)$,
\item[$\rightarrow$] with probability $\mu$, its orientation is chosen with the realization of a random variable having a probability density $g(\cdot; H^i(\nu_k),\nu_k)$.
\end{itemize}
\item protrusion number $i$ disappears with probability $\frac{d}{\overline{\tau}}$.
\item with probability $\frac{\overline{\tau} - \tau(t_{k+1})}{\overline{\tau}}$, nothing happens. This rejection event allows to recover the non-homogeneous jumping rate of our process from the dominating homogeneous process. 
\end{itemize} 
The measure $\nu_{k+1}$ is then obtained from $\nu_k$ and the information of the event occuring at time $t_{k+1}$.
\item[Updates] the particle's new position is 
\begin{displaymath}
X_{k+1} = X_k + \Delta t \, V_k\,,
\end{displaymath}
while $V_{k+1}= \frac{1}{\gamma} \begin{pmatrix}
<\nu_{k+1},\cos>\\<\nu_{k+1},\sin>
\end{pmatrix}$.
\end{description}

One only has to start again to get a trajectory over time. 

\subsection{Results}
Let us now present the numerical trajectories we obtained. In all the simulations, we used the following parameters: $T=200$, $\Delta t=10^{-4}$, $c=r=3$, $d=1$, $\gamma=200$, $\mu=0.2$, $\lambda=60$, and $\beta=2$. For the mutation kernel, we use a concentration parameter $K=50$, ensuring that mutation events are not sufficient to lose complete polarisation.\par

\subsubsection{Growing sensitivity to the signal}
In this section, we consider an empty initial protrusion population. We choose $$\alpha\in \{ 0.01,\, 0.1,\, 1,\,10\},$$ and for each configuration we display several trajectories for an increasing sensitivity to the signal $\kappa_2$. 
We also show for $\alpha=0.01$ and $\alpha=10$ the histograms of velocity module and orientation.

Since $\beta=2$ is the maximal sensitivity to the internal persistance, we consider $\kappa_2\in \{0.1,\,1,\,2\}$. In the last case, this means that the signal intensity equally competes with the inner sensitivity. Figure \ref{fig:noCI} shows the corresponding cell trajectories, while figure \ref{fig:noCI_distrib} shows the numerical distributions of the velocity module and orientation for 200 trajectories lasting $T=100$ each, with $\dx t = 10^{-3}$.


We make several observations. First, as in previous works, we can see that when $\alpha$ increases, trajectories are more persistent and explore a larger territory. This can also be seen in the velocity module diagram, where larger values are attained for growing $\alpha$.\par 
 The case $\kappa_2=0.1$ corresponds to a very low sensitivity to the signal, so that (blue) trajectories are not biased in that direction. The upper pannel in the velocity diagram in figure \ref{fig:noCI_distrib} also shows that the velocity orientations are not biased.\par 
 The case of an intermediate sensitivity ($\kappa_2=1$, red trajectories and middle panel in figure \ref{fig:noCI_distrib}) shows a bifurcation when $\alpha$ grows: the sensitivity to self polarisation is larger, so that for $\alpha=0.01$, the isotropic distribution wins over the asymmetry induced by the signal: the cell does not move in the direction of the gradient. However, when the self-polarisation rises, the reproduction events get more concentrated, and a preferred direction appears. Then, even if other directions than the one of the gradient may be followed, the constant bias towards the right has a visible effect on the migration. \par 
Finally, for cells as sensitive to the signal than to themselves (yellow trajectories and lower panel), the preferred orientation is always clear, correlating to very persistent trajectories. 

\begin{figure}[H]
	\centering
	\subfloat[$\alpha=0.01$]{\includegraphics[scale=0.15]{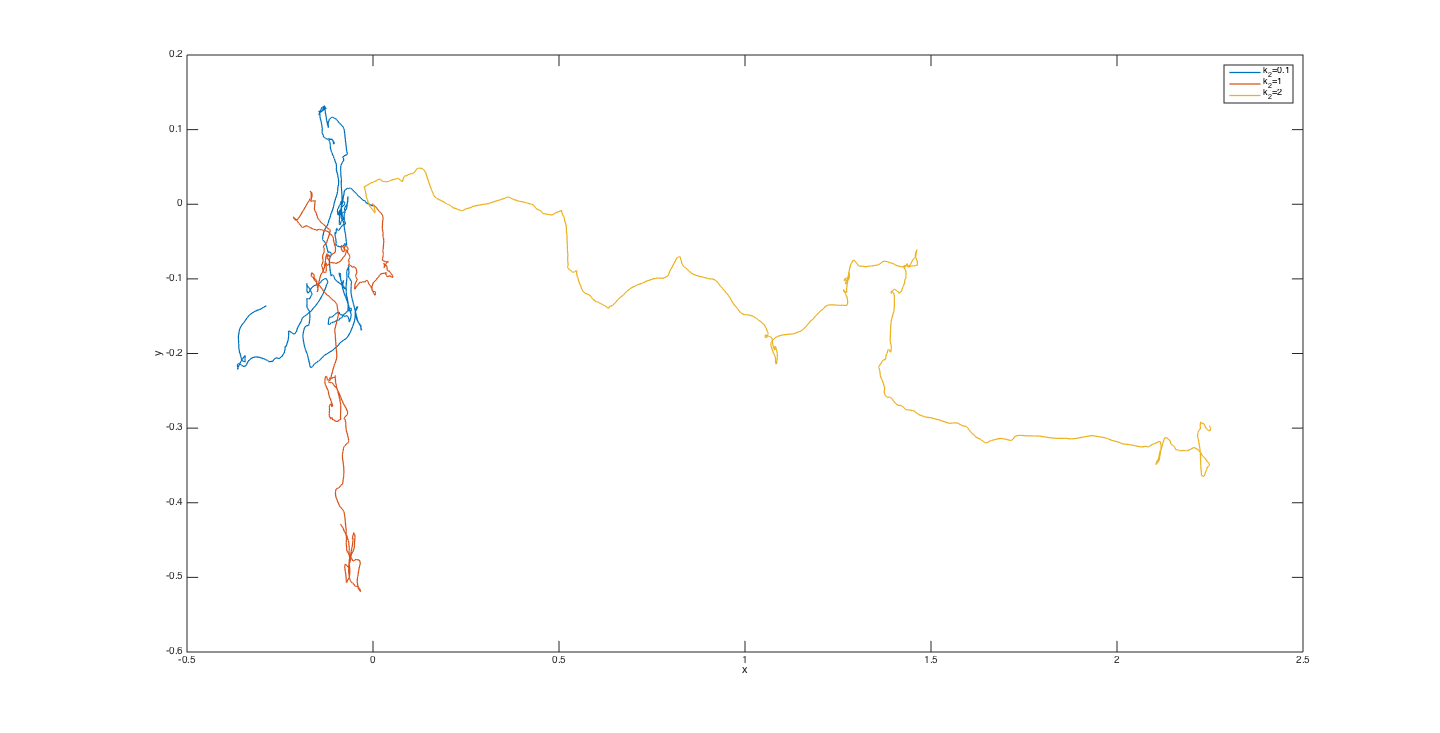}}\quad
	\subfloat[$\alpha=0.1$]{\includegraphics[scale=0.15]{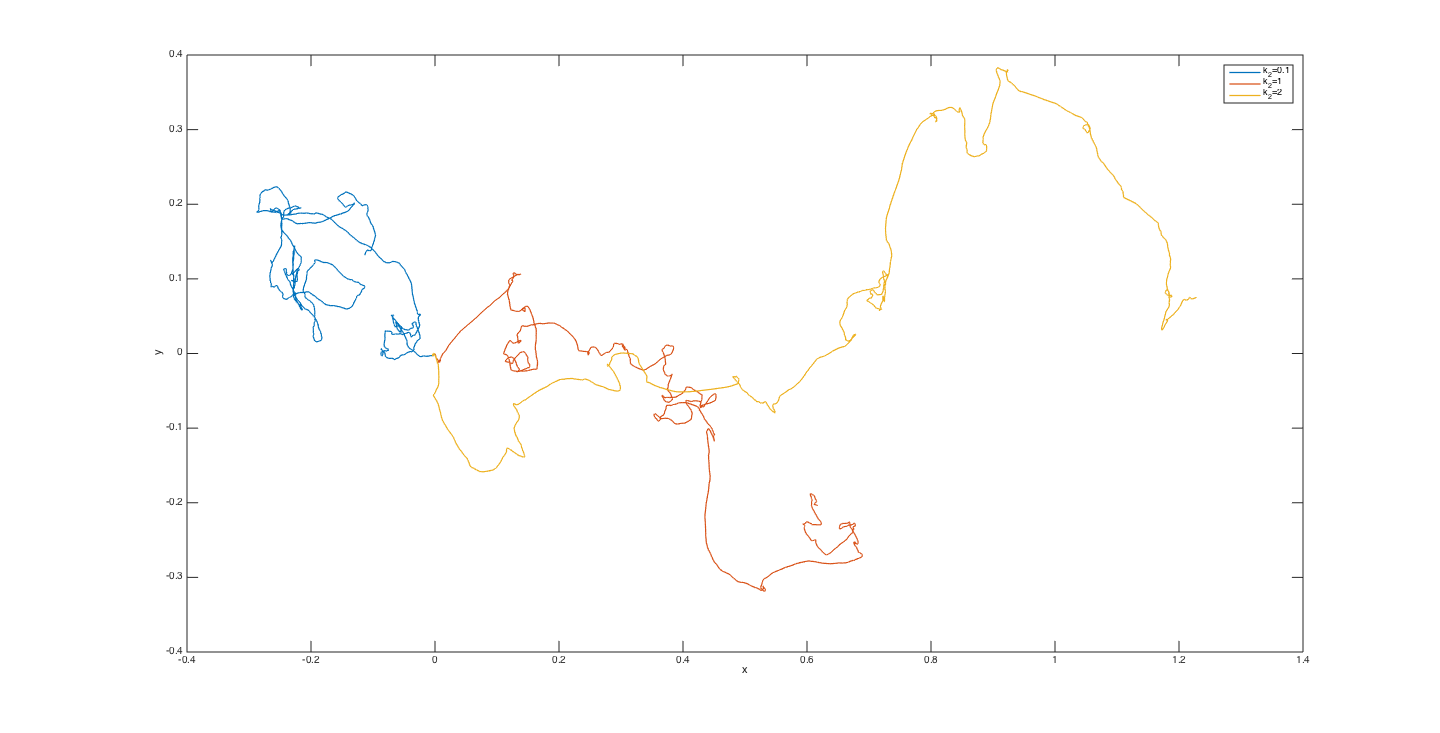}}\\	
		\subfloat[$\alpha=1$]{\includegraphics[scale=0.15]{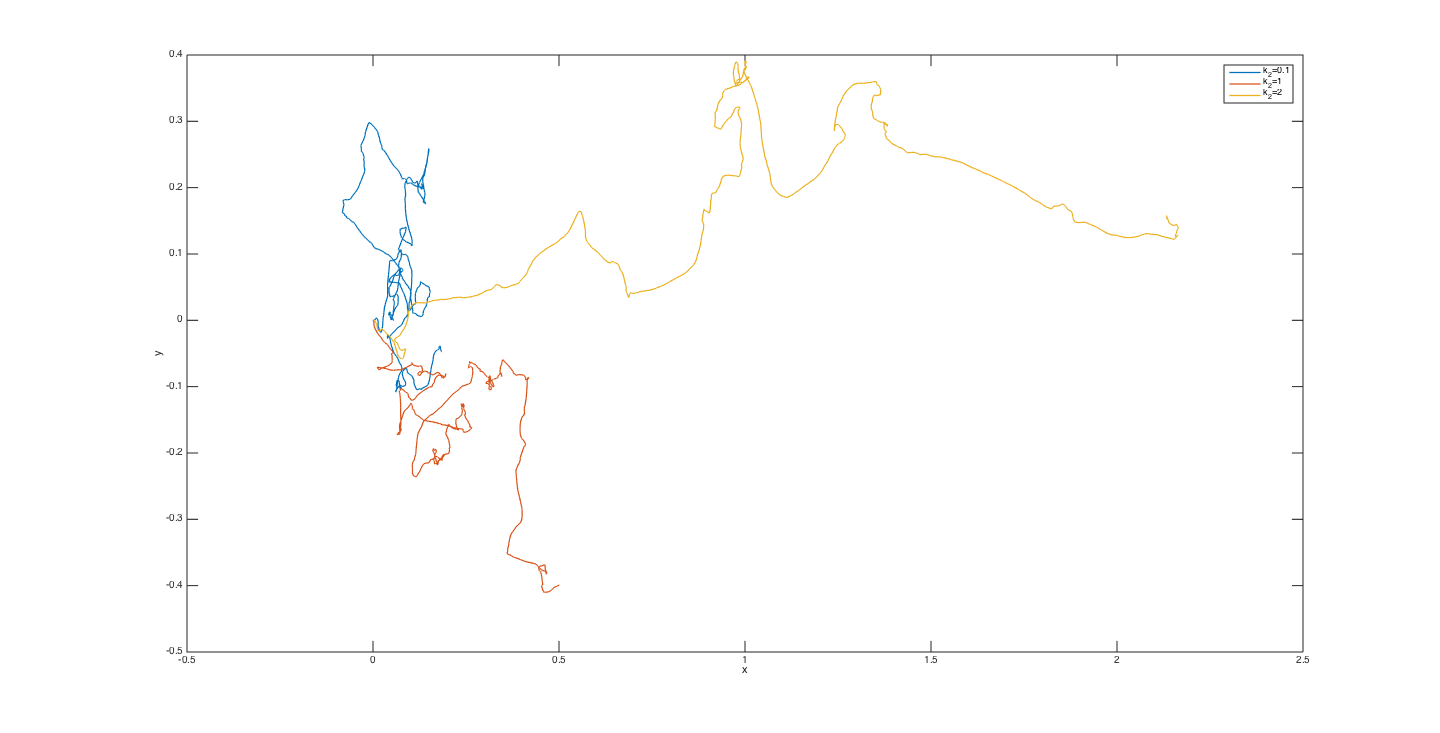}}\quad
	\subfloat[$\alpha=10$]{\includegraphics[scale=0.15]{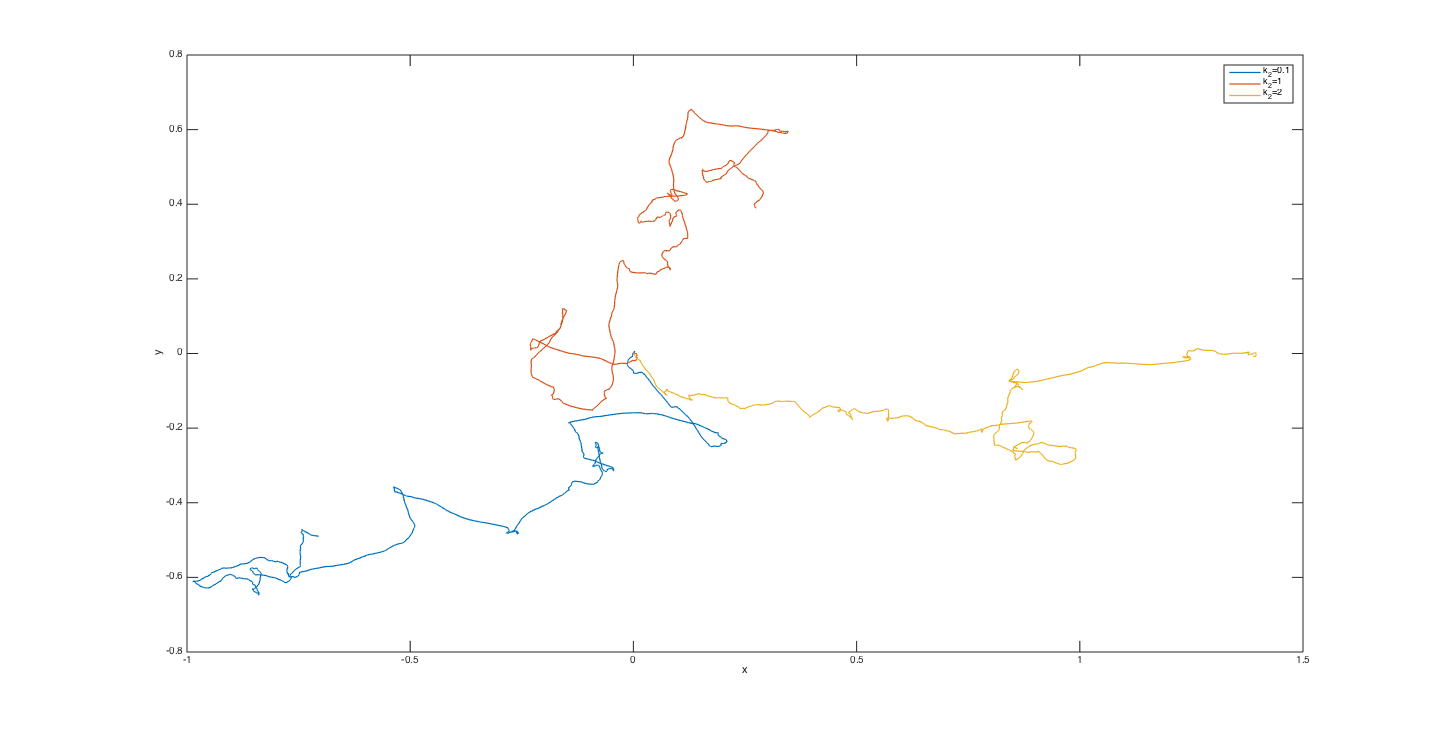}}\\	
	\caption{Numerical trajectories for empty initial protrusion population and several values for $\alpha$ and $\kappa_2$.}\label{fig:noCI}
\end{figure}

These simulations illustrate the balance between the inner and outer sensing of the cell, and shows the variety of behaviours arising from it. In the following, we investigate the role of the initial condition on the cell behaviour.

\begin{figure}[H]
\centering
\subfloat[$\alpha=0.01$]{\includegraphics[scale=0.5]{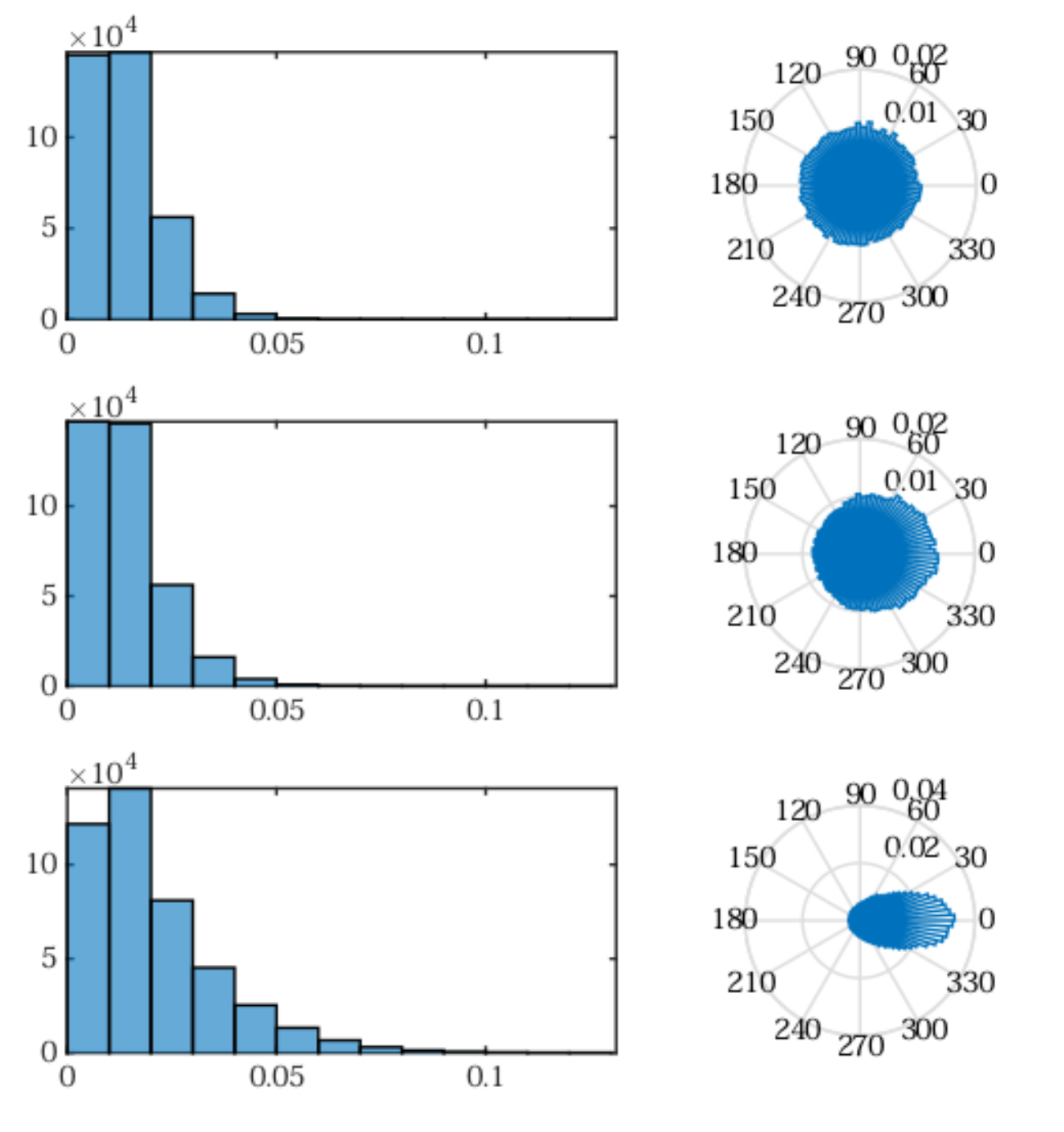}}\quad
\subfloat[$\alpha=10$]{\includegraphics[scale=0.5]{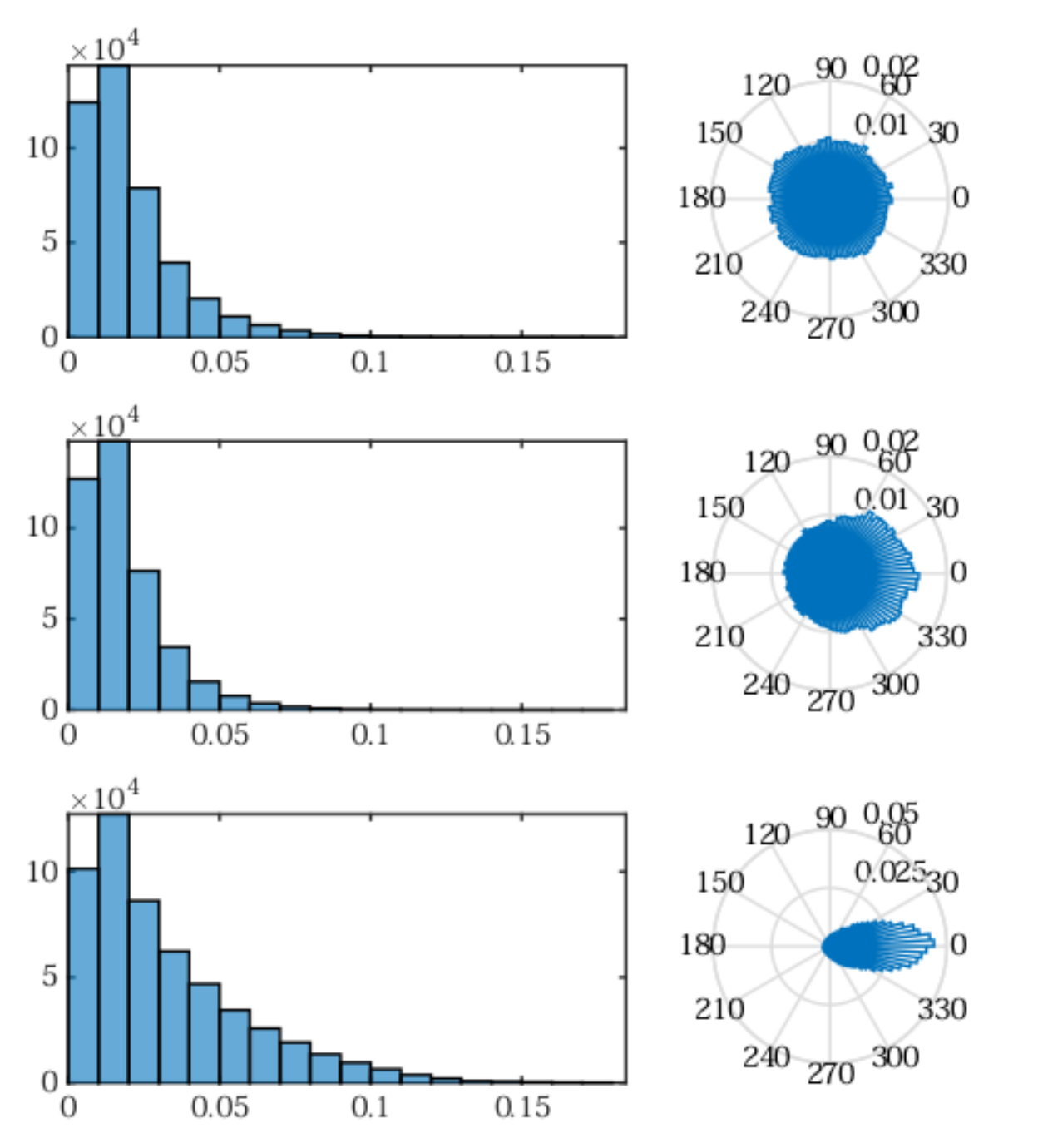}}
\caption{Distributions of the velocity module and orientation for $\alpha=0.01$ (left) and $\alpha=10$ (right) obtained for 200 realizations in each case. Up: $\kappa_2=0.1$, Middle: $\kappa_2=1$, Down: $\kappa_2=2$. Parameters: $T=100$, $\dx t = 10^{-3}$.}\label{fig:noCI_distrib}
\end{figure}

\subsubsection{Inverted initial polarisation}
We perform now the same numerical simulations from a polarised initial condition against the direction of the gradient of signal. More precisely, the initial population of protrusions is composed of $5$ protrusions oriented in the $\pi$ direction. The corresponding trajectories are displayed in figure \ref{fig:CI}, and histograms of velocity modules and orientations for $\alpha\in \{ 0.01,\,10\}$ are displayed in figure \ref{fig:CI_distrib}.

\begin{figure}[H]
	\centering
	\subfloat[$\alpha=0.01$]{\includegraphics[scale=0.15]{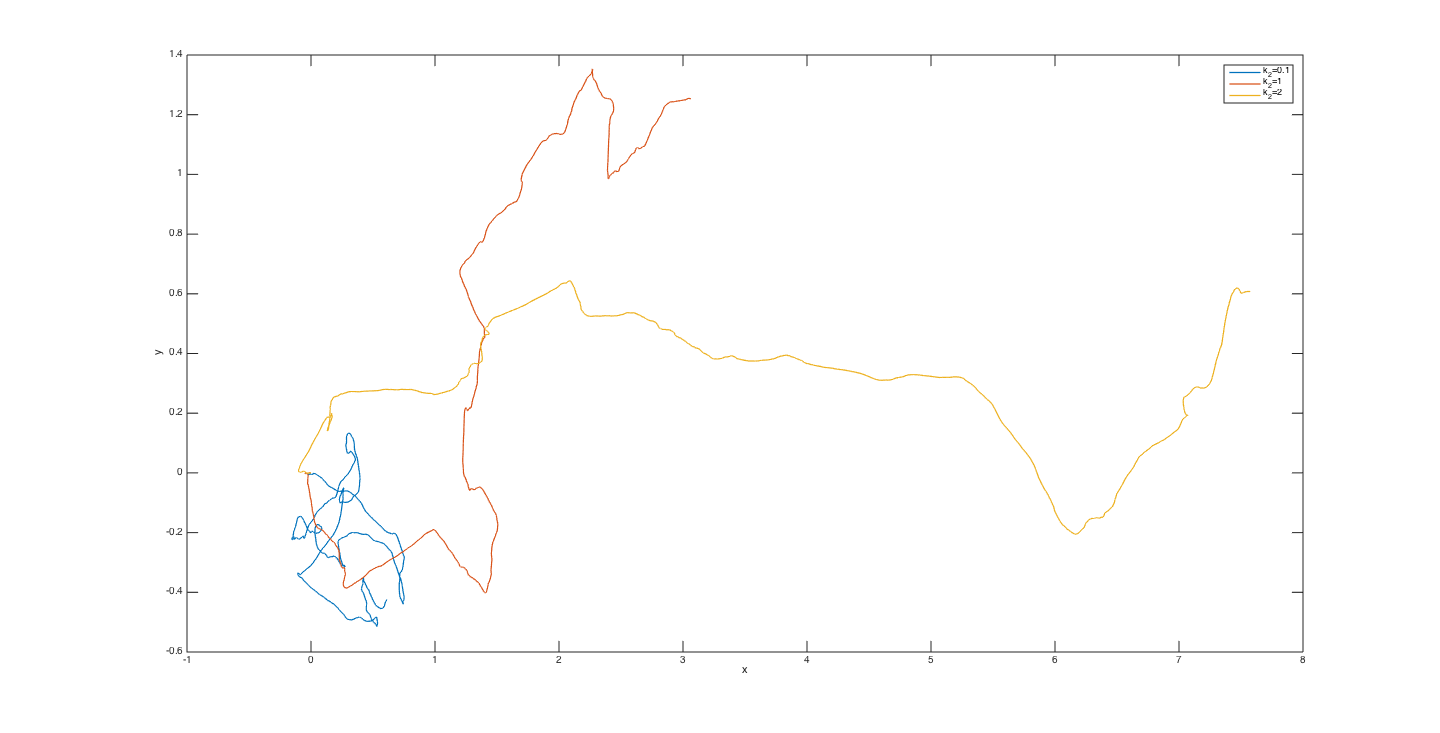}}\quad
	\subfloat[$\alpha=0.1$]{\includegraphics[scale=0.15]{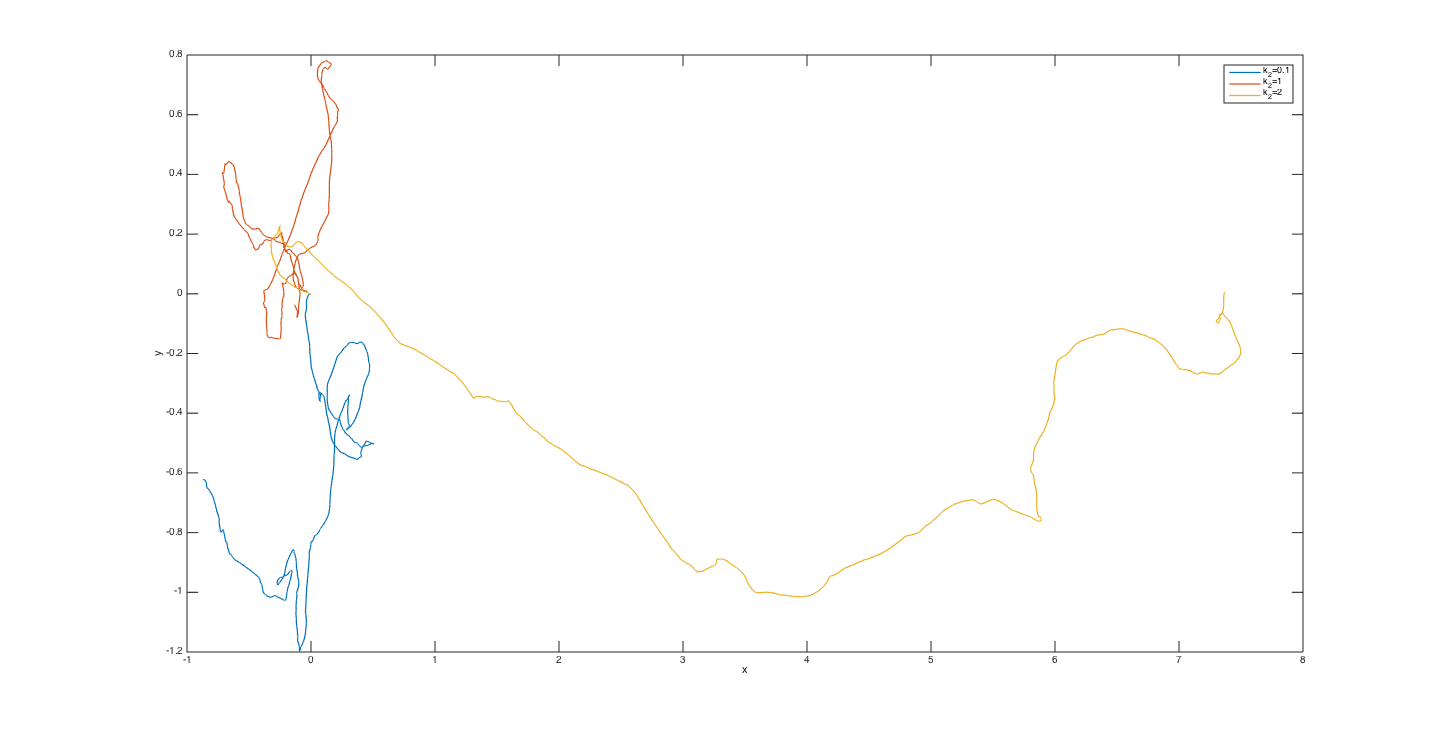}}\\	
		\subfloat[$\alpha=1$]{\includegraphics[scale=0.15]{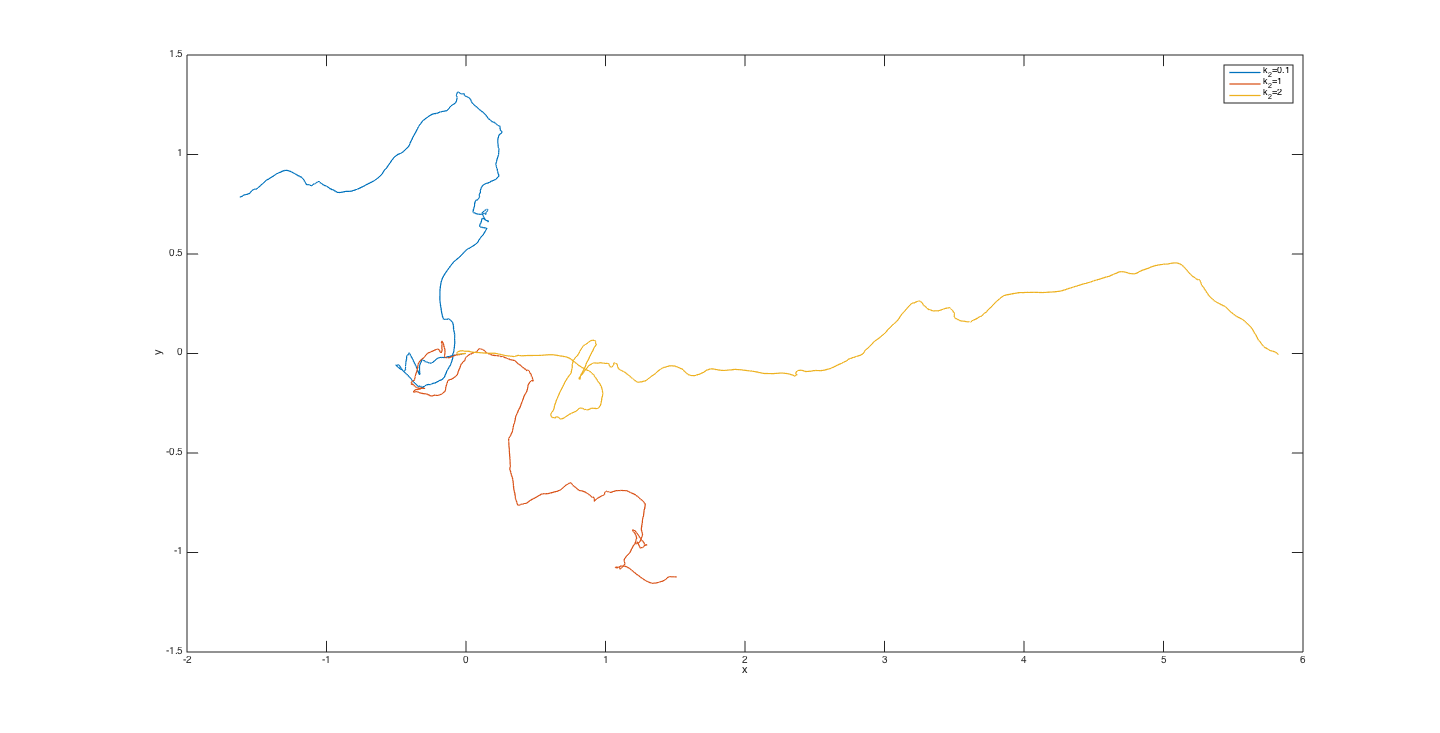}}\quad
	\subfloat[$\alpha=10$]{\includegraphics[scale=0.15]{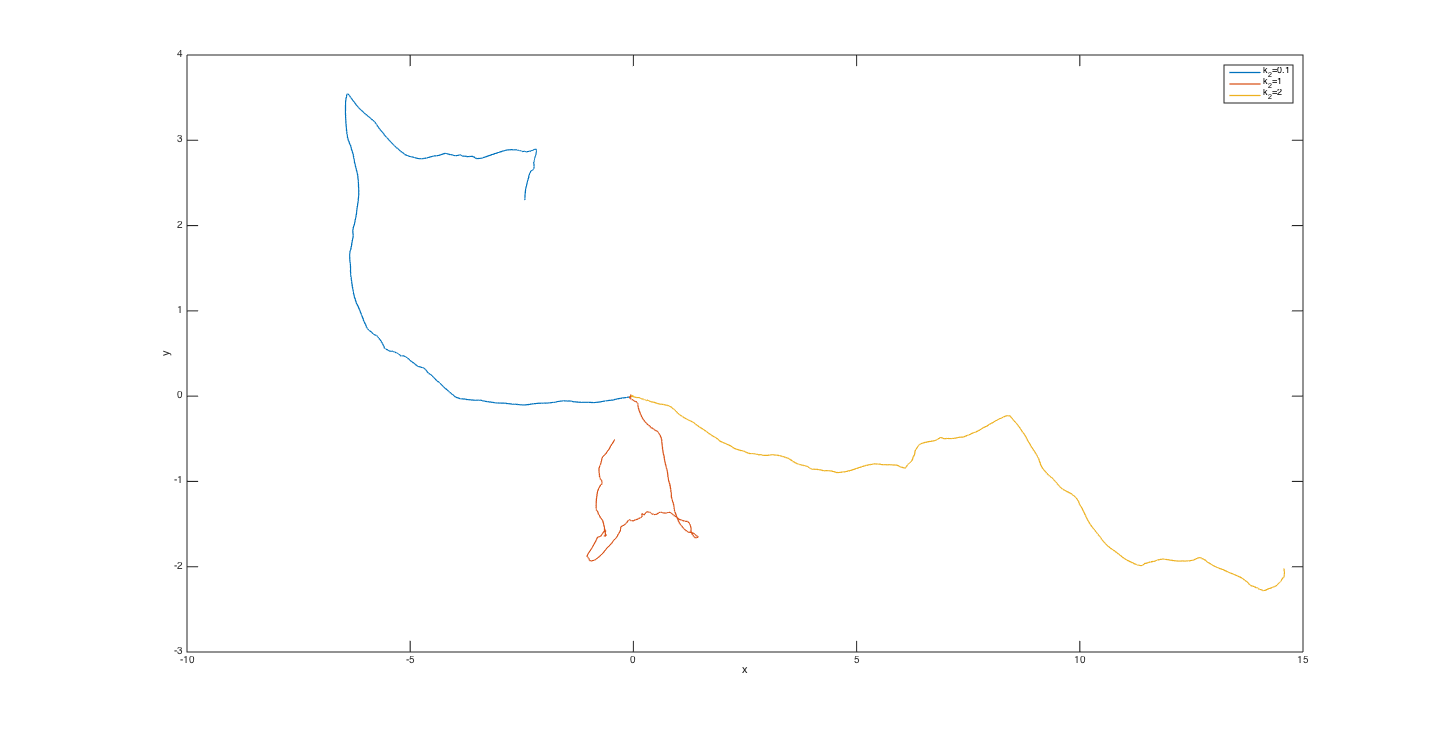}}\\	
	\caption{Numerical trajectories for a $\pi$-polarised initial condition and several values for $\alpha$ and $\kappa_2$.}\label{fig:CI}
\end{figure}

\begin{figure}[H]
\centering
\subfloat[$\alpha=0.01$]{\includegraphics[scale=0.5]{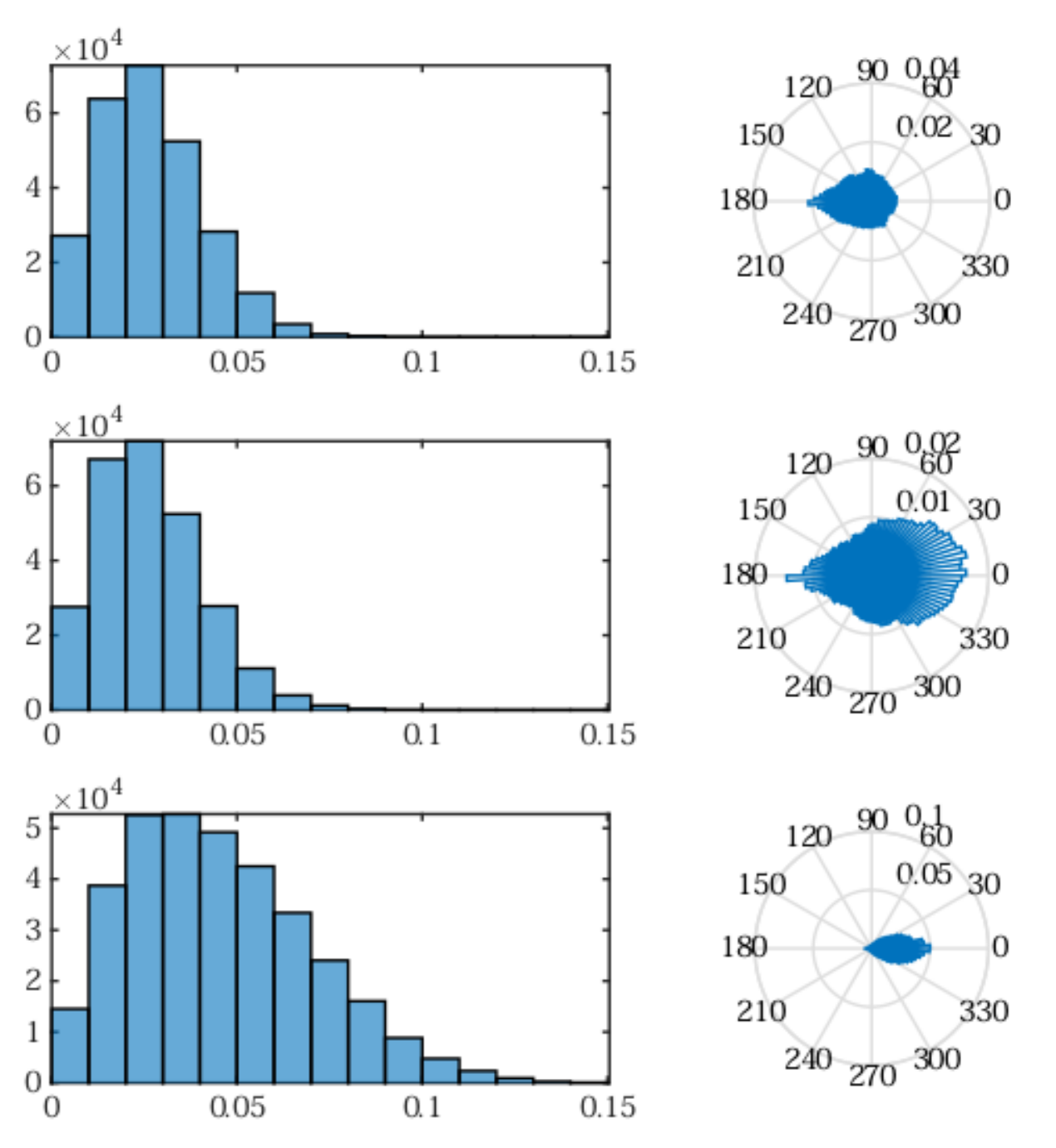}}\qquad
\subfloat[$\alpha=10$]{\includegraphics[scale=0.5]{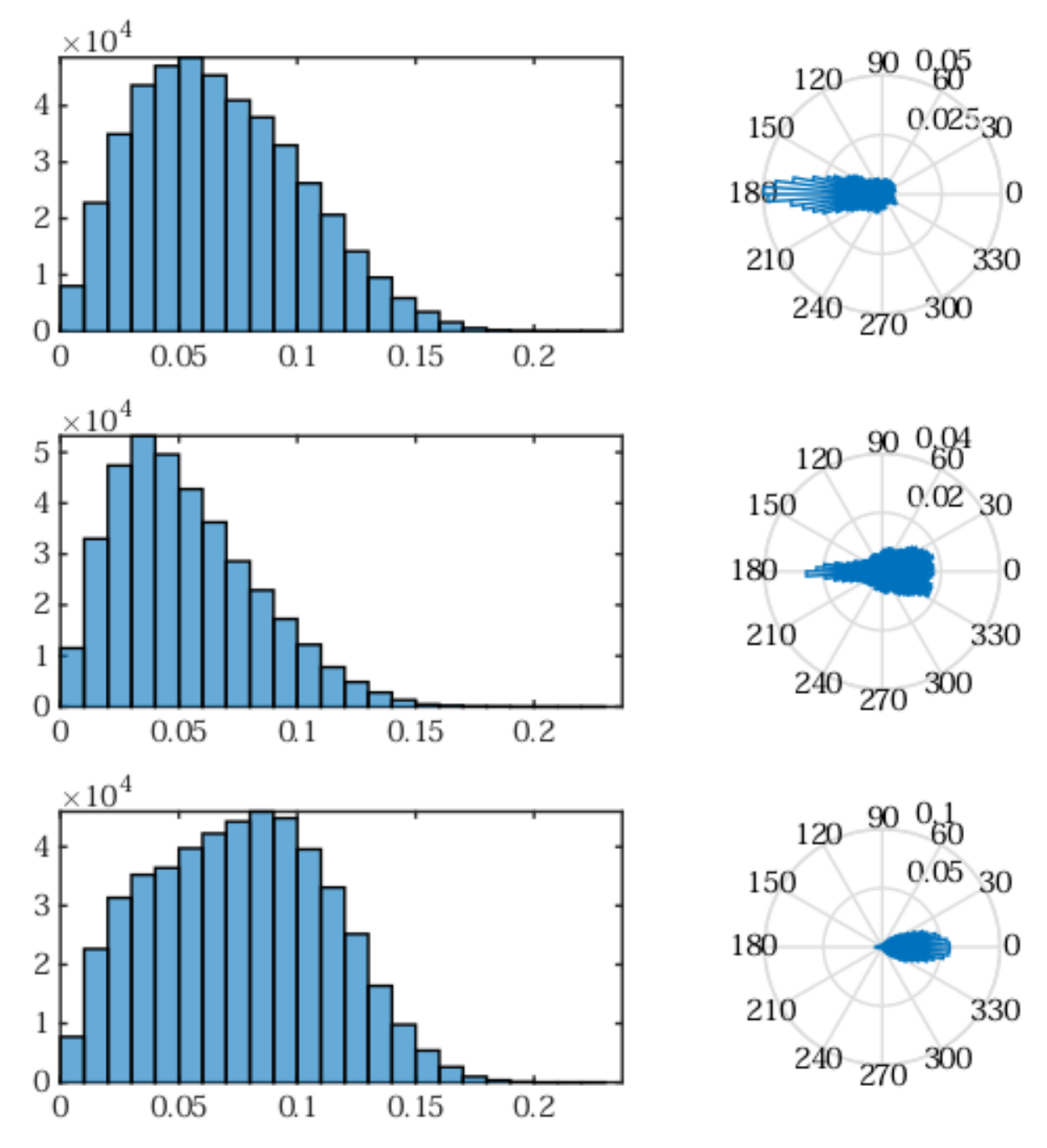}}
\caption{Distributions of the velocity module and orientation for $\alpha=0.01$ (left) and $\alpha=10$ (right) obtained for 200 realizations in each case, for a polarized initial condition. Up: $\kappa_2=0.1$, Middle: $\kappa_2=1$, Down: $\kappa_2=2$. Parameters: $T=50$, $\dx t = 10^{-3}$.}\label{fig:CI_distrib}
\end{figure}

We notice first that starting with an initial polarization increases neatly the polarization and the cells' velocity modules. The effect of a large $\alpha$ is therefore more visible. For lower values of $\alpha$, the initial polarisation has little effet and we observe the same behaviours as before. However, note that the trajectories are smoother due to the larger number of protrusions arising from the reproduction. \par 
Finally, intermediate values of $\kappa_2$ (lower than $\beta$) show that the self-enhanced cell machinery can play against the signal.

	\subsubsection{Time-dependent gradient of signal}
	We consider now time-dependent gradient of signal. 
	
	\paragraph{Single On-Off}
	We first explore the the case of an effective signal during half of the experiment time interval, before being put to zero until the end. The resulting trajectories are displayed in figure \ref{fig:OnOff}.
	
	\begin{figure}[H]
	\centering
	\subfloat[$\alpha=0.01$]{\includegraphics[scale=0.15]{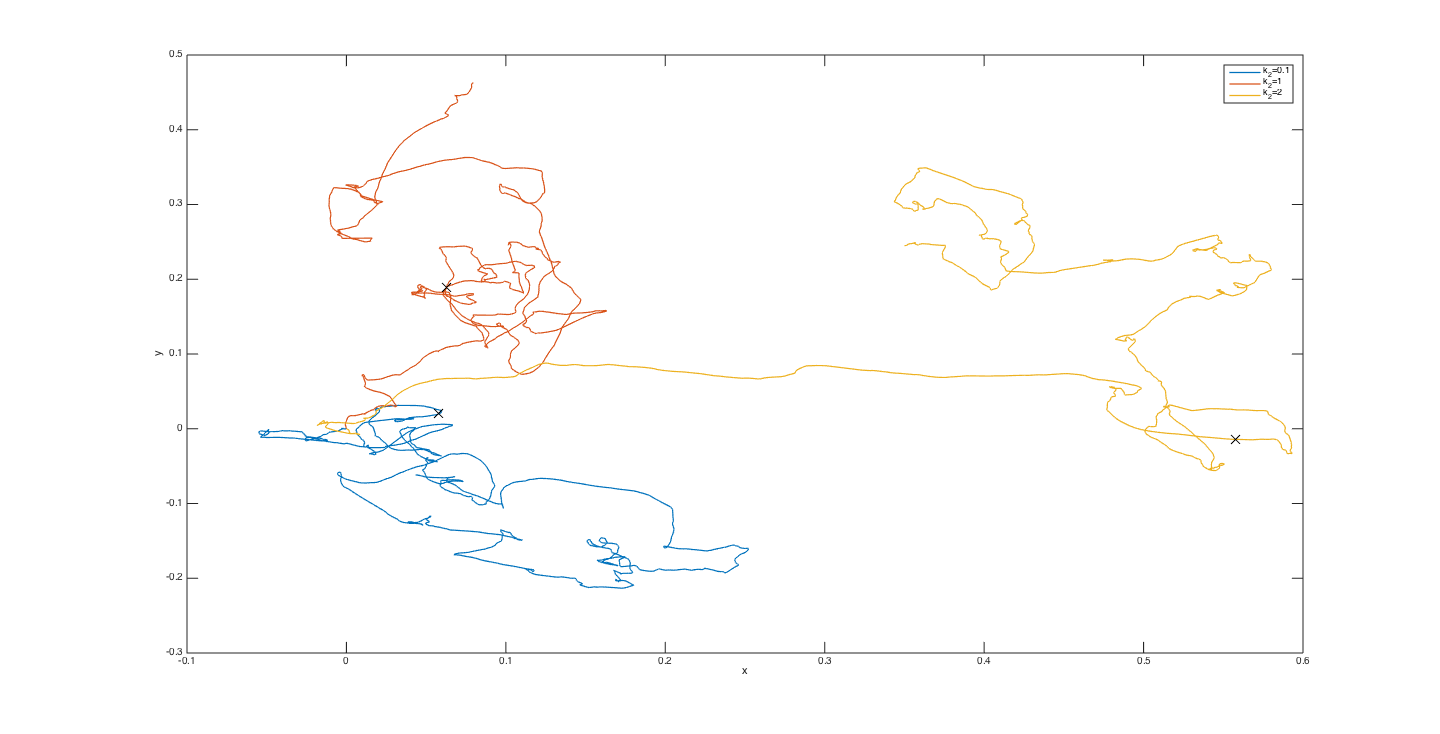}}\quad
	\subfloat[$\alpha=0.1$]{\includegraphics[scale=0.15]{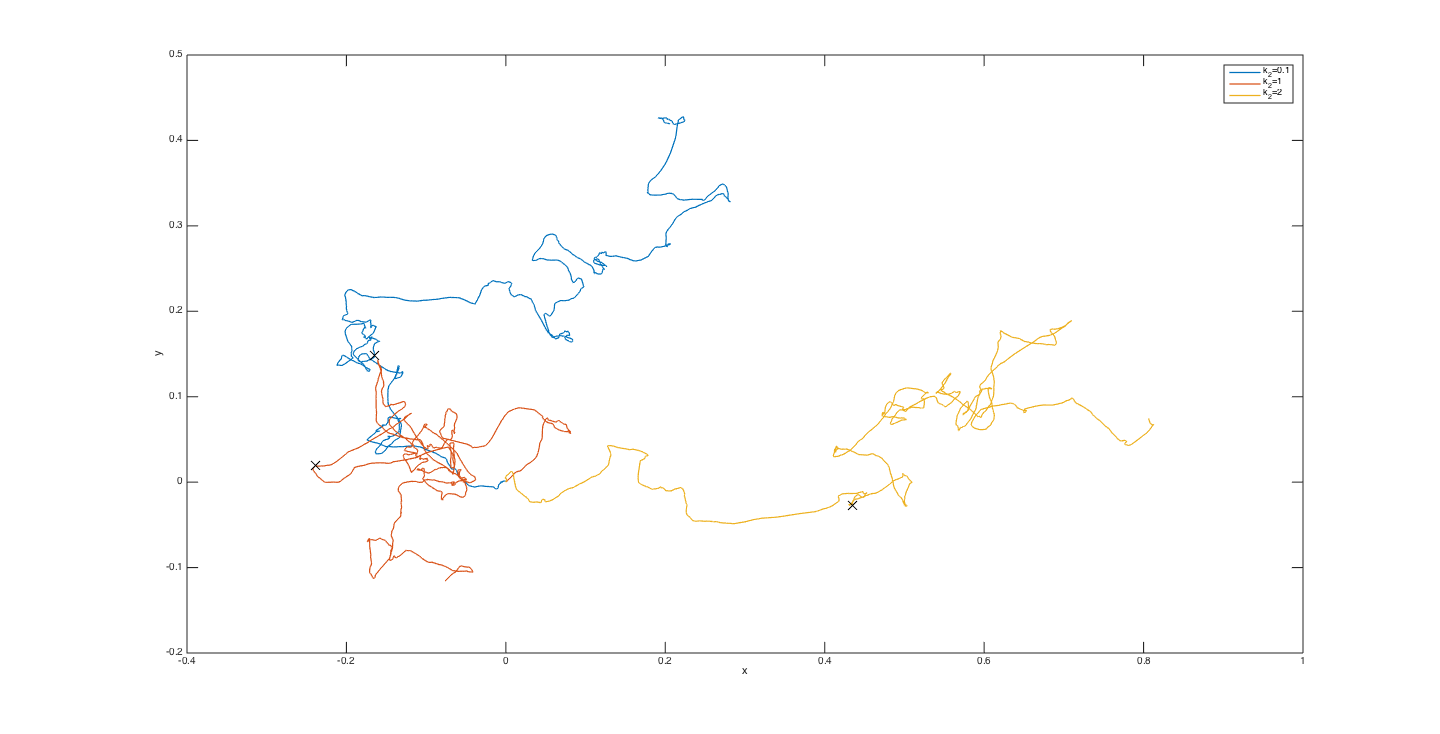}}\\	
		\subfloat[$\alpha=1$]{\includegraphics[scale=0.15]{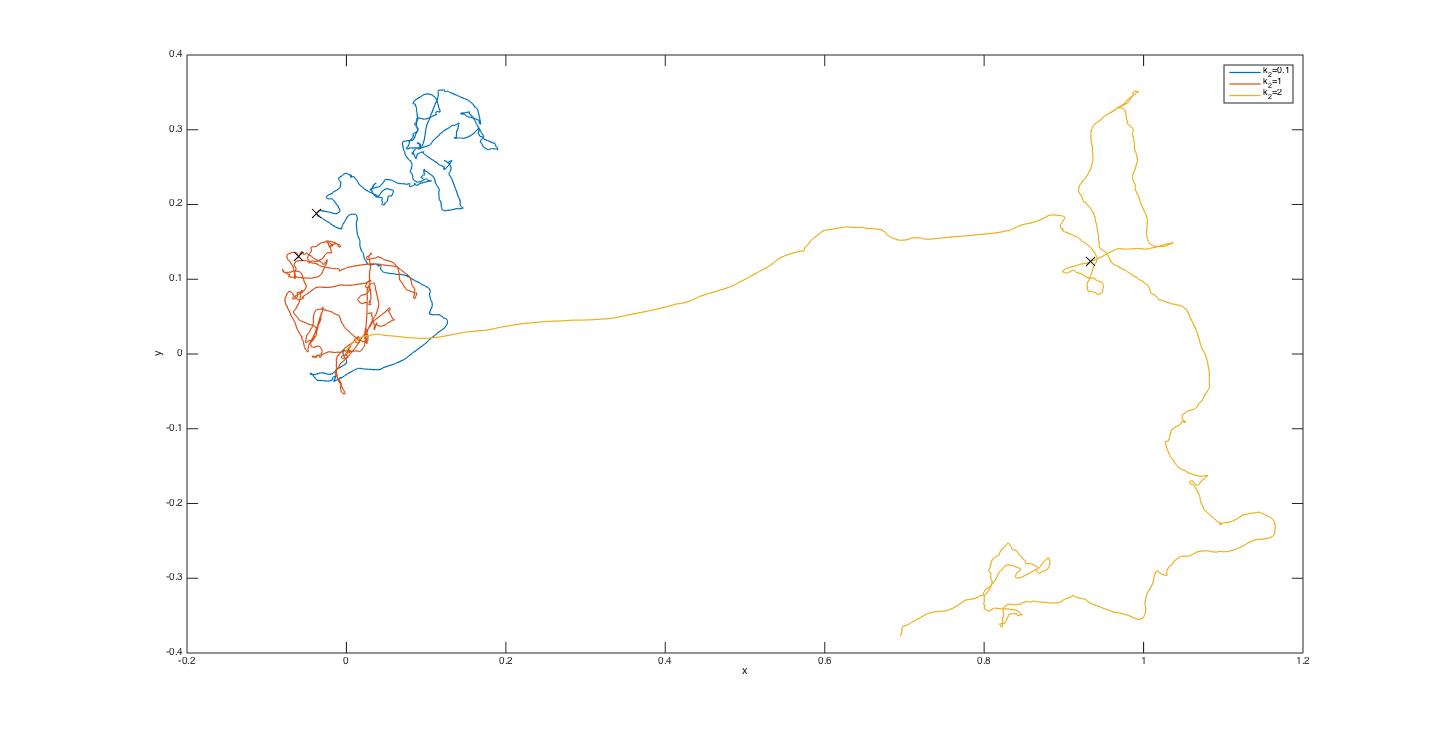}}\quad
	\subfloat[$\alpha=10$]{\includegraphics[scale=0.15]{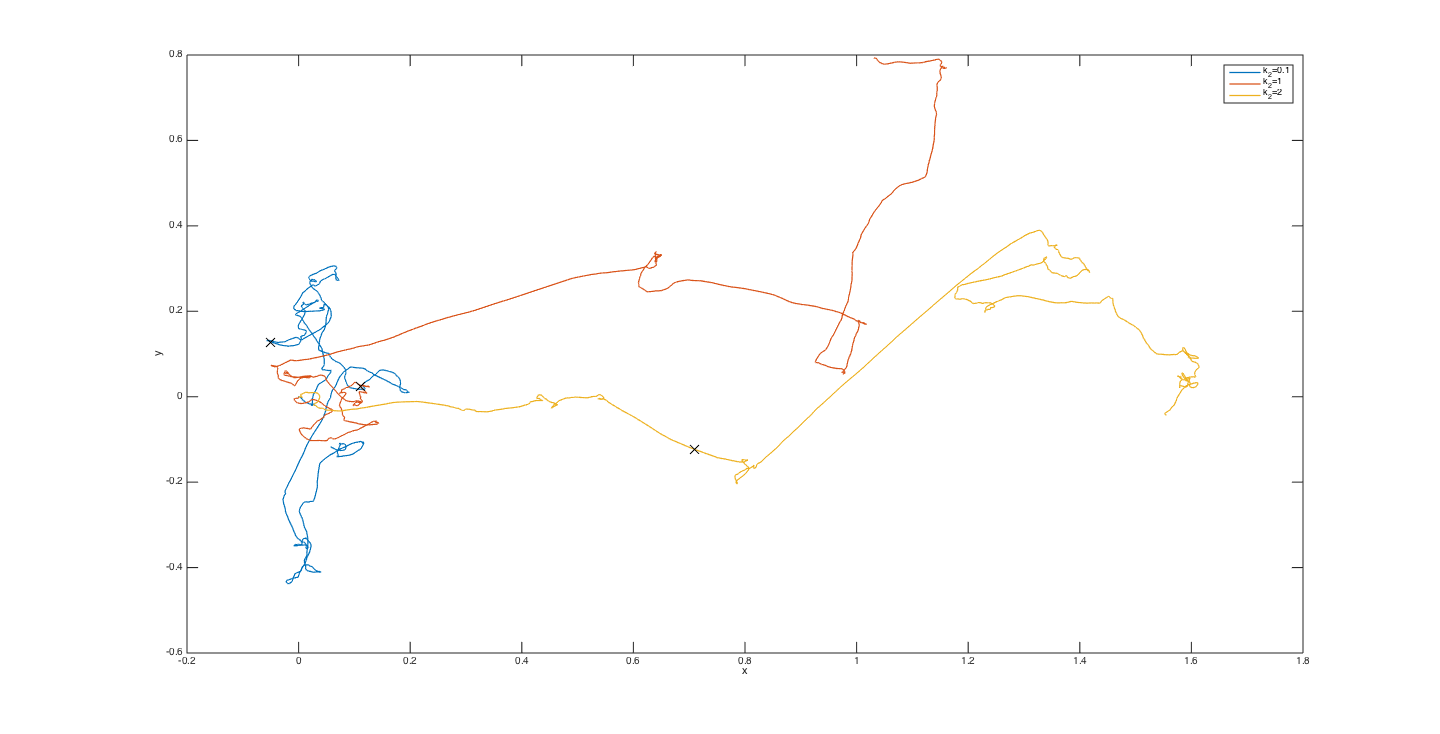}}\\	
	\caption{Numerical trajectories in a On-Off gradient of signal, for a non polarised initial condition and several values for $\alpha$ and $\kappa_2$. The $X$ symbols marks the time where the signal is switched off.}\label{fig:OnOff}
\end{figure}

From these simulations, we can make two observations. First, for low polarised cells ($\kappa_2=0.1$ and sometimes $\kappa_2=1$), if the first chosen direction is opposed to the one of the gradient, both phenomena tend to cancel each other, so that the trajectories are not persistent. Then, when the signal switches off, a better efficiency is recovered. 
If the cell is initially polarised in the direction of the gradient, then the opposite occurs: both phenomena cooperate and the persistence is lost when the signal is switched off. For the more sensitive cells, the initial absence of motion is not able to blur the signal, so that trajectories follow the gradient while it exists.


\paragraph{Periodic On-Off}
Then, the same kind of experiment is led for cycles of signalling: in figure \ref{fig:OnOff10h}, the signal is switched on or off every $10 \si{\hour}$, and in figure \ref{fig:OnOff1h} every $1\si{\hour}$. The resulting trajectories show some intermittent behaviours.  

\begin{figure}[H]
	\centering
	\subfloat[$\alpha=0.01$]{\includegraphics[scale=0.15]{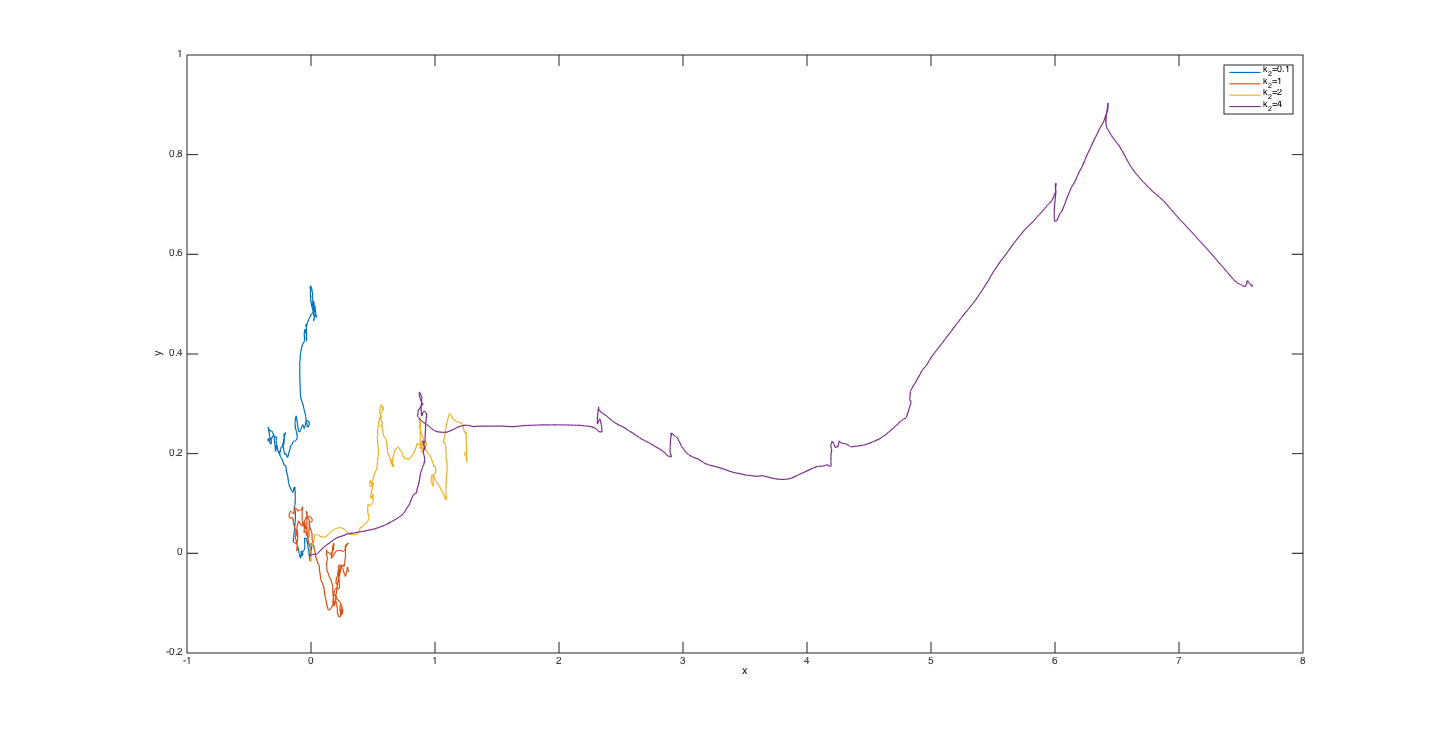}}\quad
	\subfloat[$\alpha=0.1$]{\includegraphics[scale=0.15]{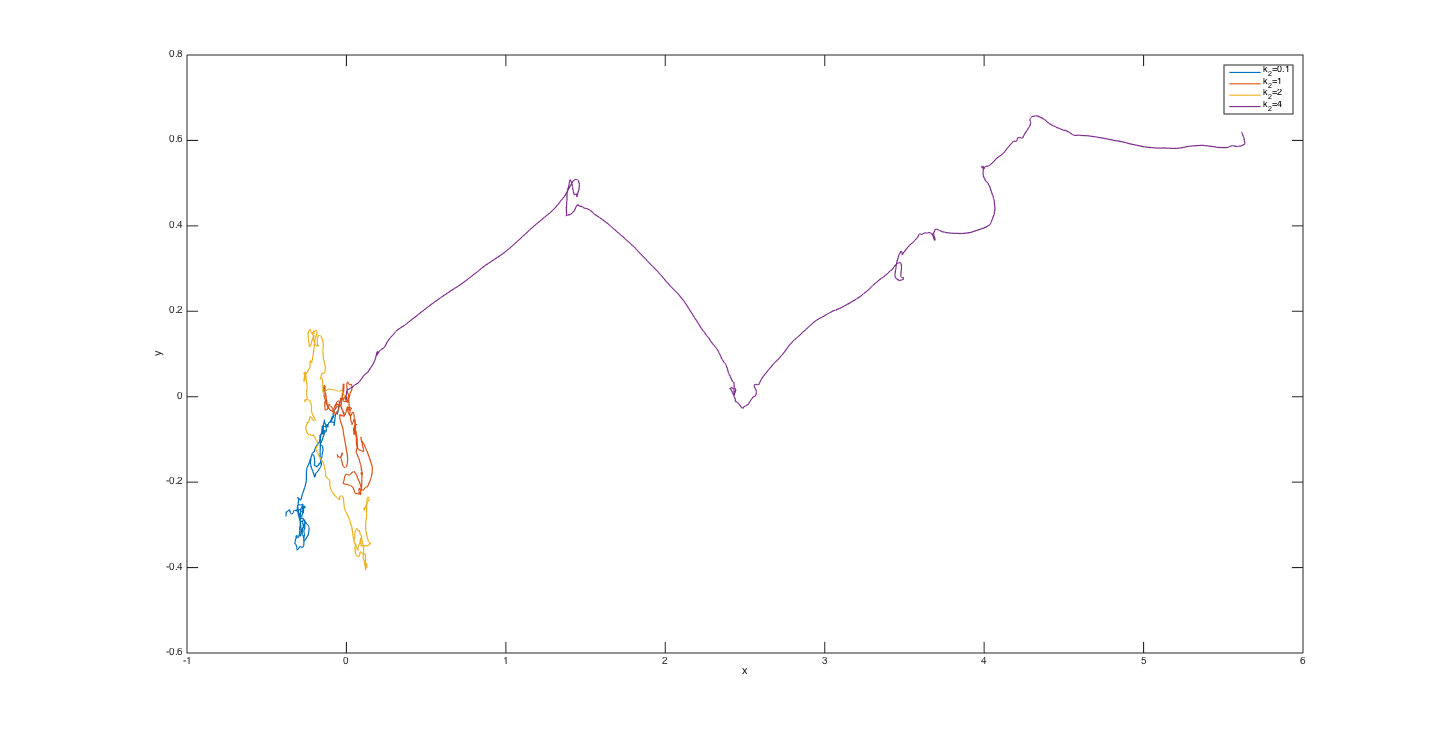}}\\	
		\subfloat[$\alpha=1$]{\includegraphics[scale=0.15]{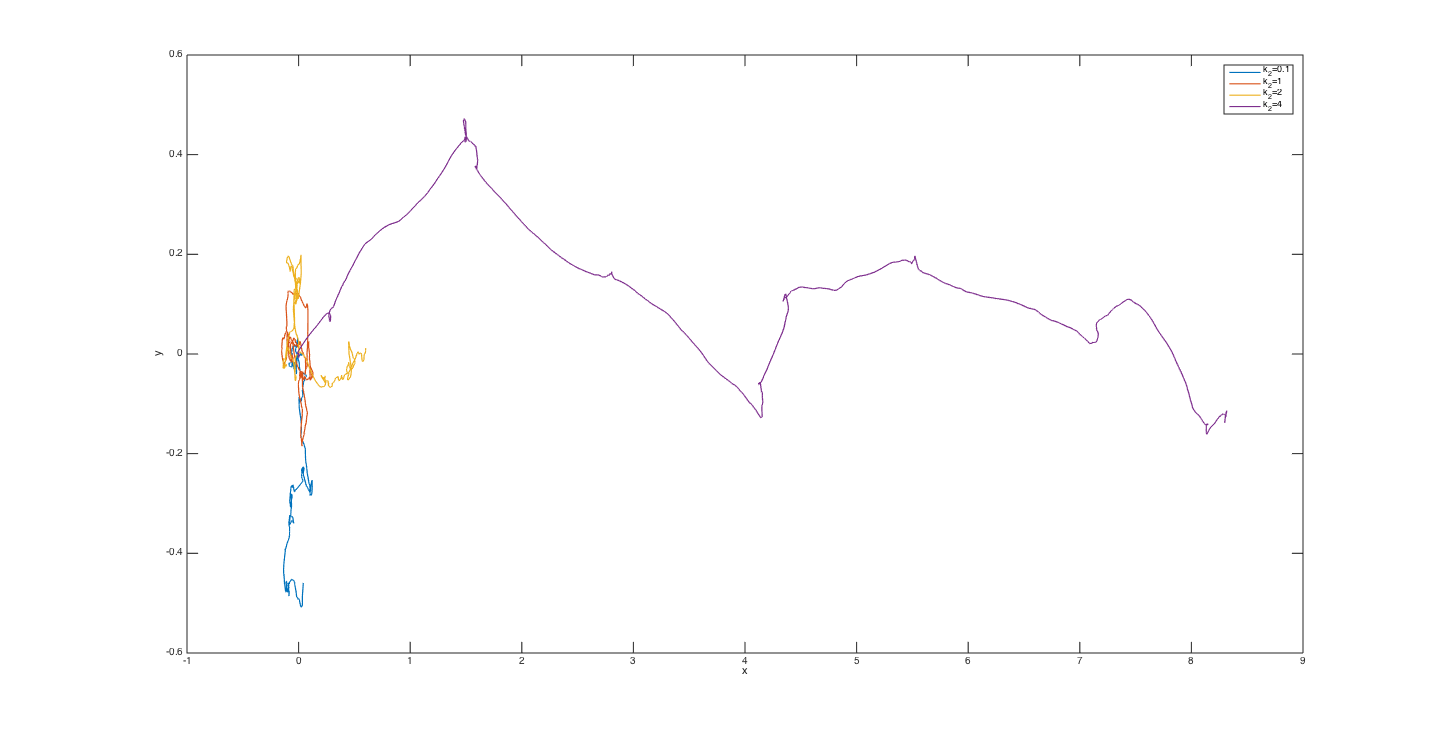}}\quad
	\subfloat[$\alpha=10$]{\includegraphics[scale=0.15]{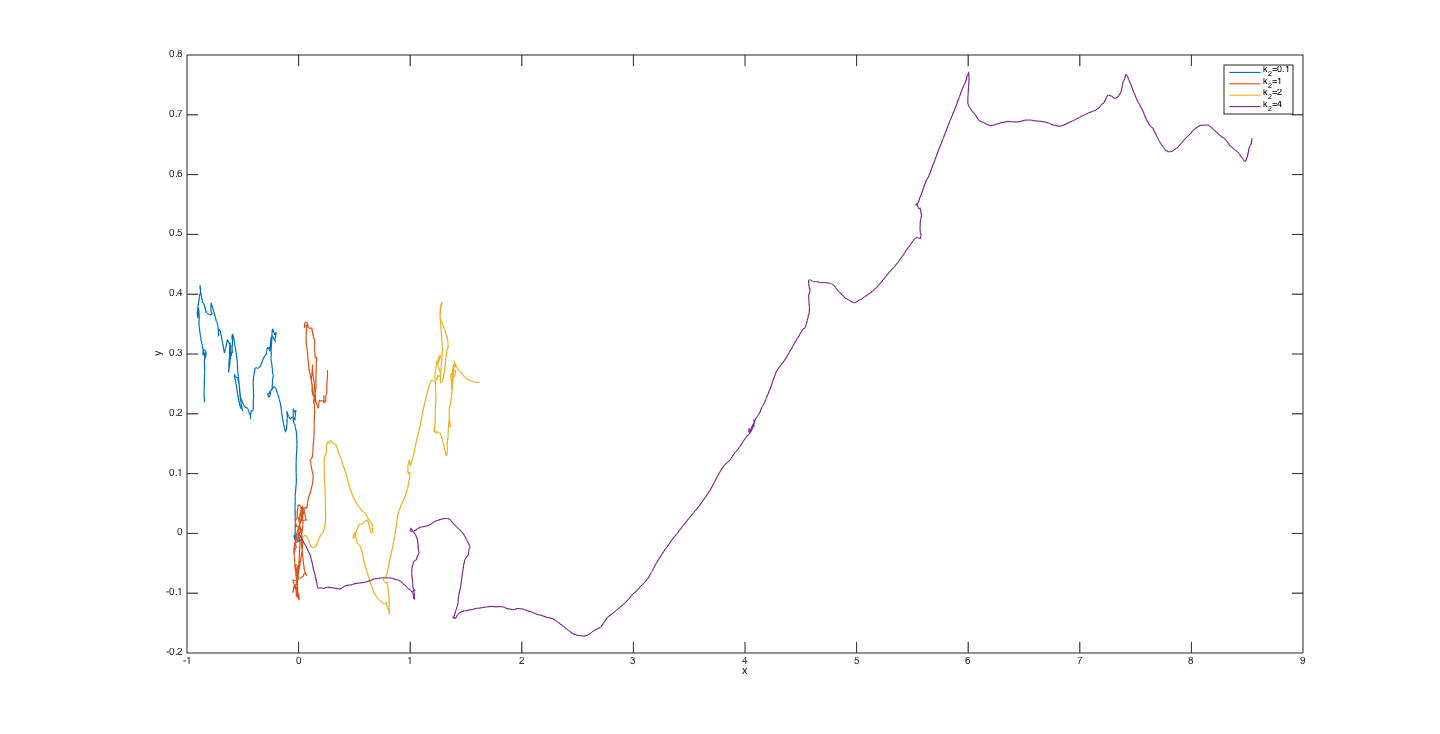}}\\	
	\caption{Numerical trajectories in a 10h On-Off setting, for a non polarised initial condition and several values for $\alpha$ and $\kappa_2$.}\label{fig:OnOff10h}
\end{figure}

\begin{figure}[H]
	\centering
	\subfloat[$\alpha=0.01$]{\includegraphics[scale=0.15]{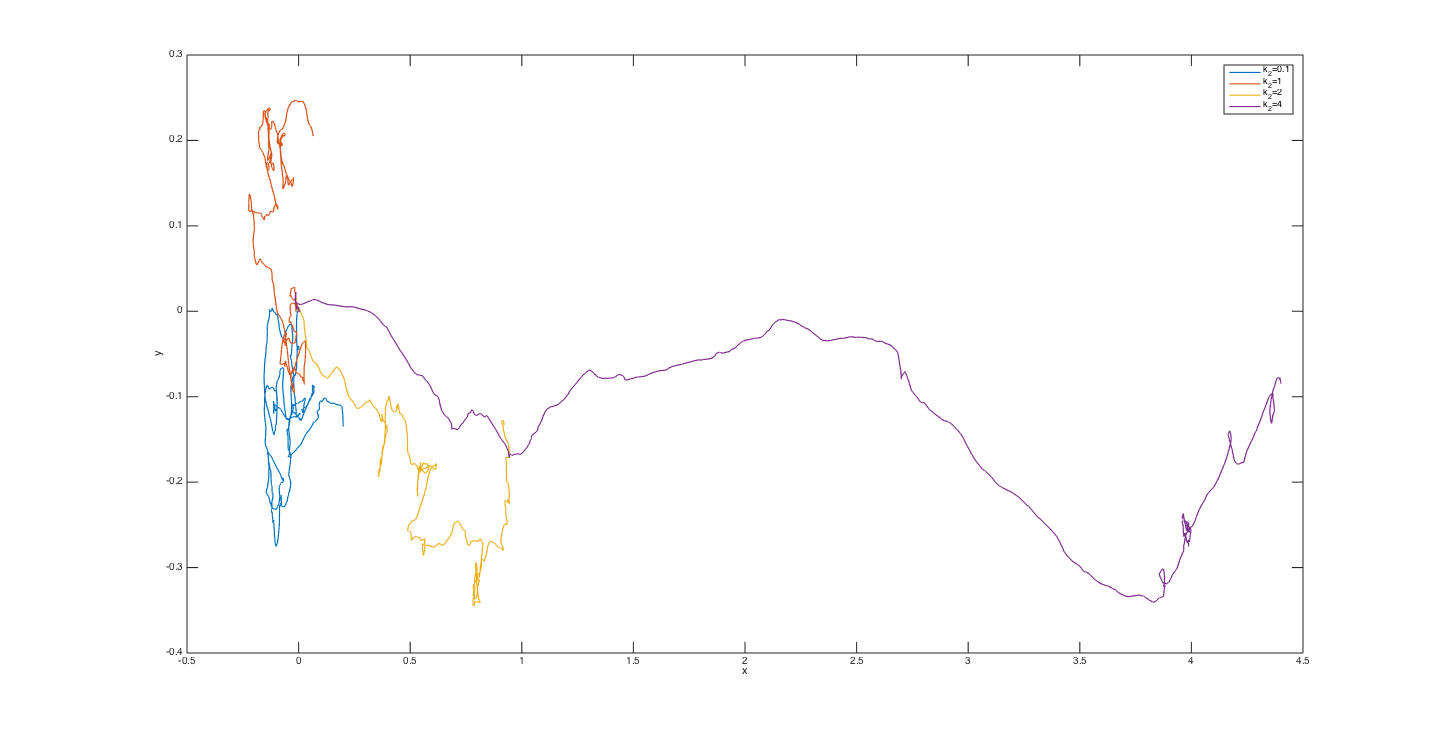}}\quad
	\subfloat[$\alpha=0.1$]{\includegraphics[scale=0.15]{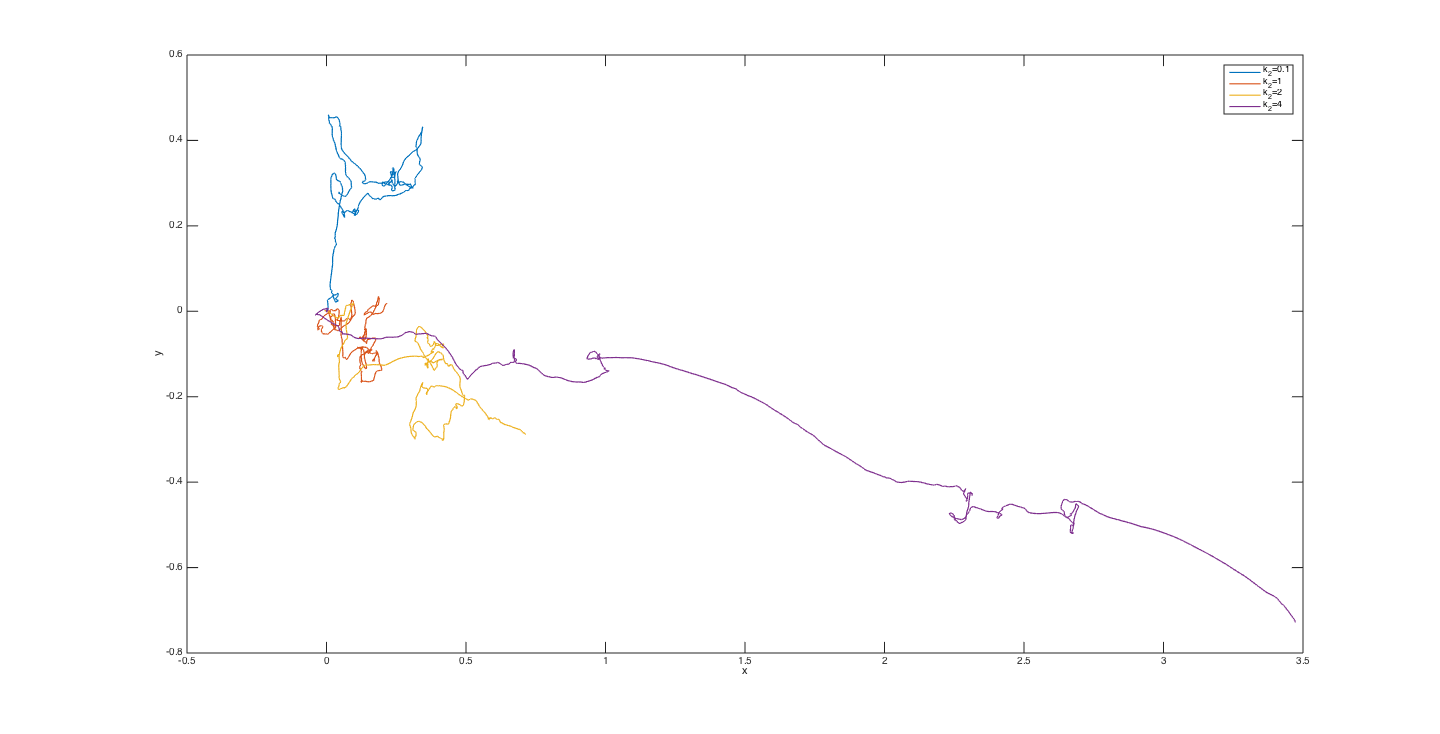}}\\	
		\subfloat[$\alpha=1$]{\includegraphics[scale=0.15]{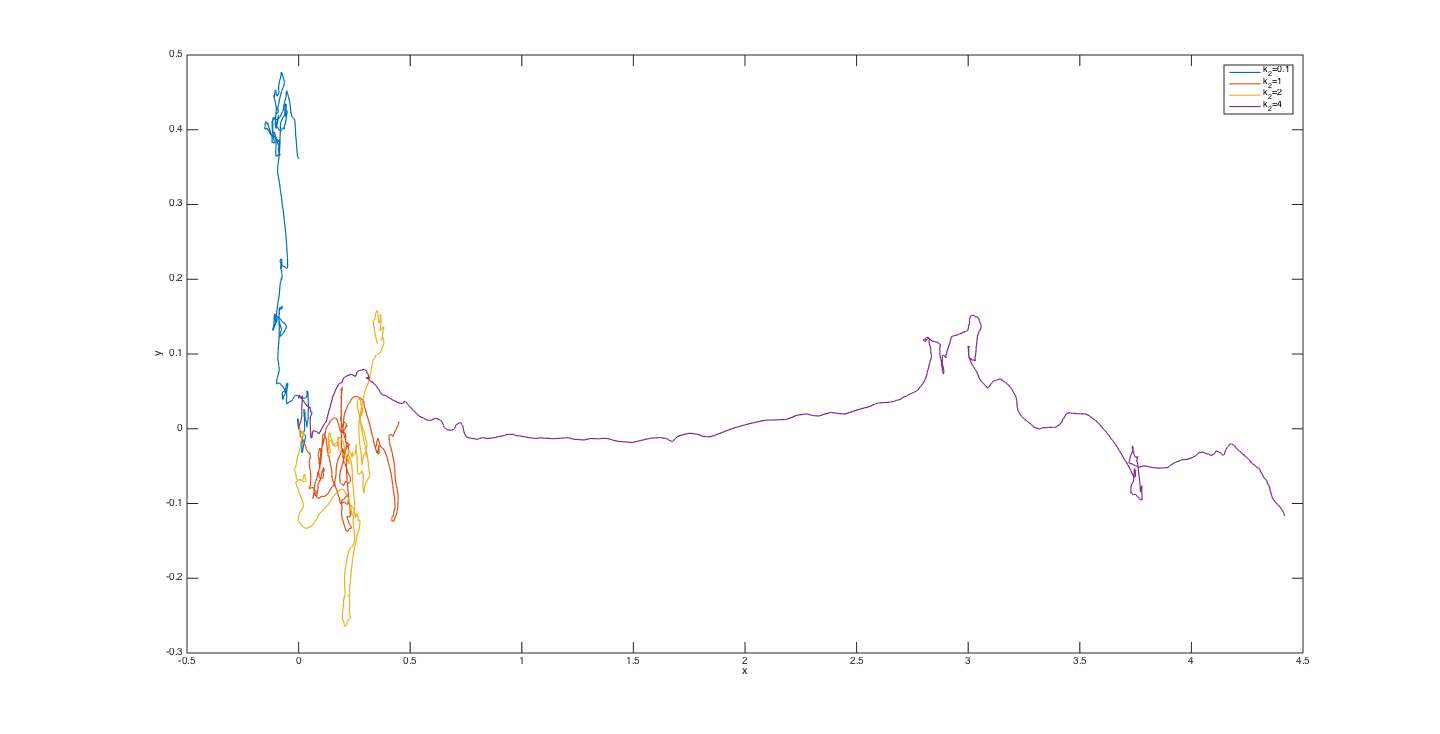}}\quad
	\subfloat[$\alpha=10$]{\includegraphics[scale=0.15]{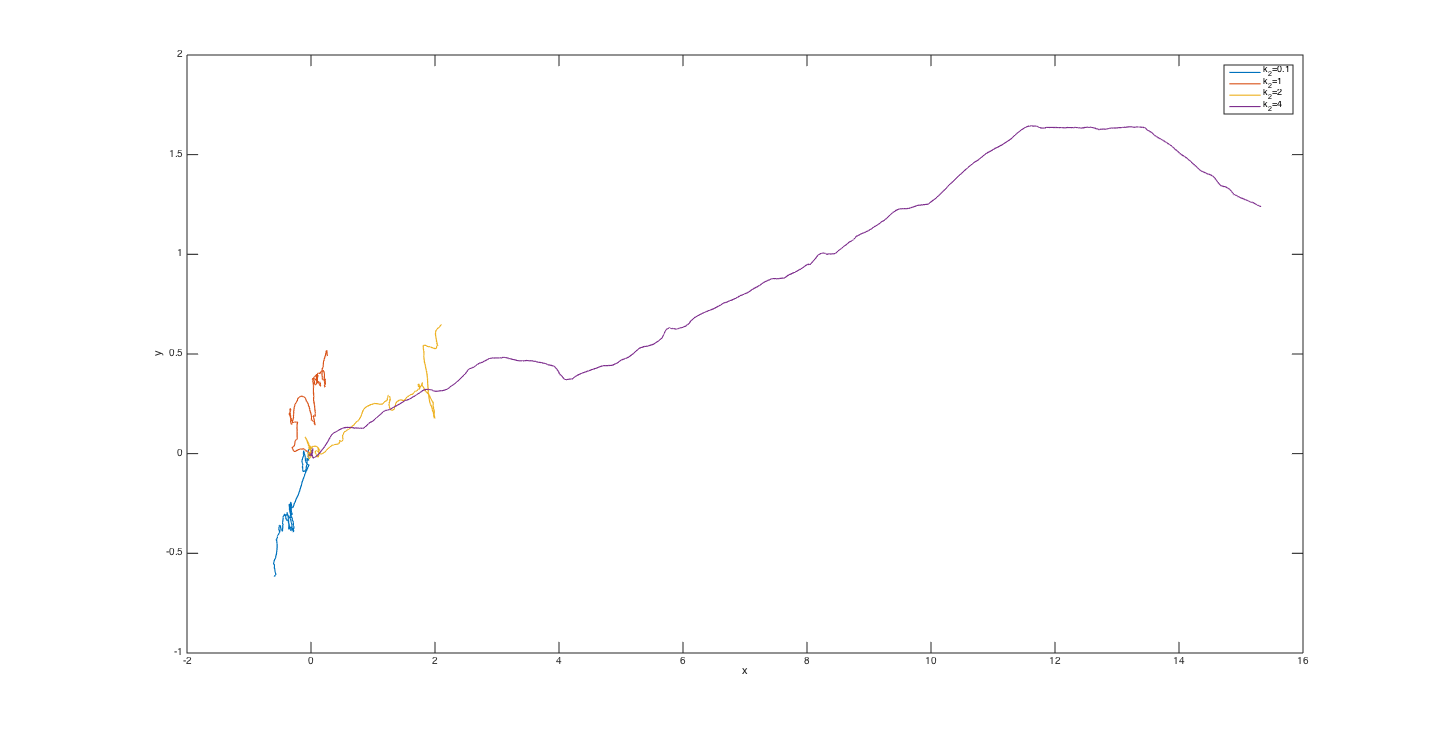}}\\	
	\caption{Numerical trajectories in a 1h On-Off setting, for a non polarised initial condition and several values for $\alpha$ and $\kappa_2$.}\label{fig:OnOff1h}
\end{figure}

\paragraph{Continuous time oscillations}
We choose now a signal sensitivity continuously oscillating around the value $k_2$. More precisely, take $\kappa_2(t,k_2)= k_2(1+\cos(t)) \in [0,2k_2]$. The trajectories are displayed in figure \ref{fig:OnOffcos}.

\begin{figure}[H]
	\centering
	\subfloat[$\alpha=0.01$]{\includegraphics[scale=0.15]{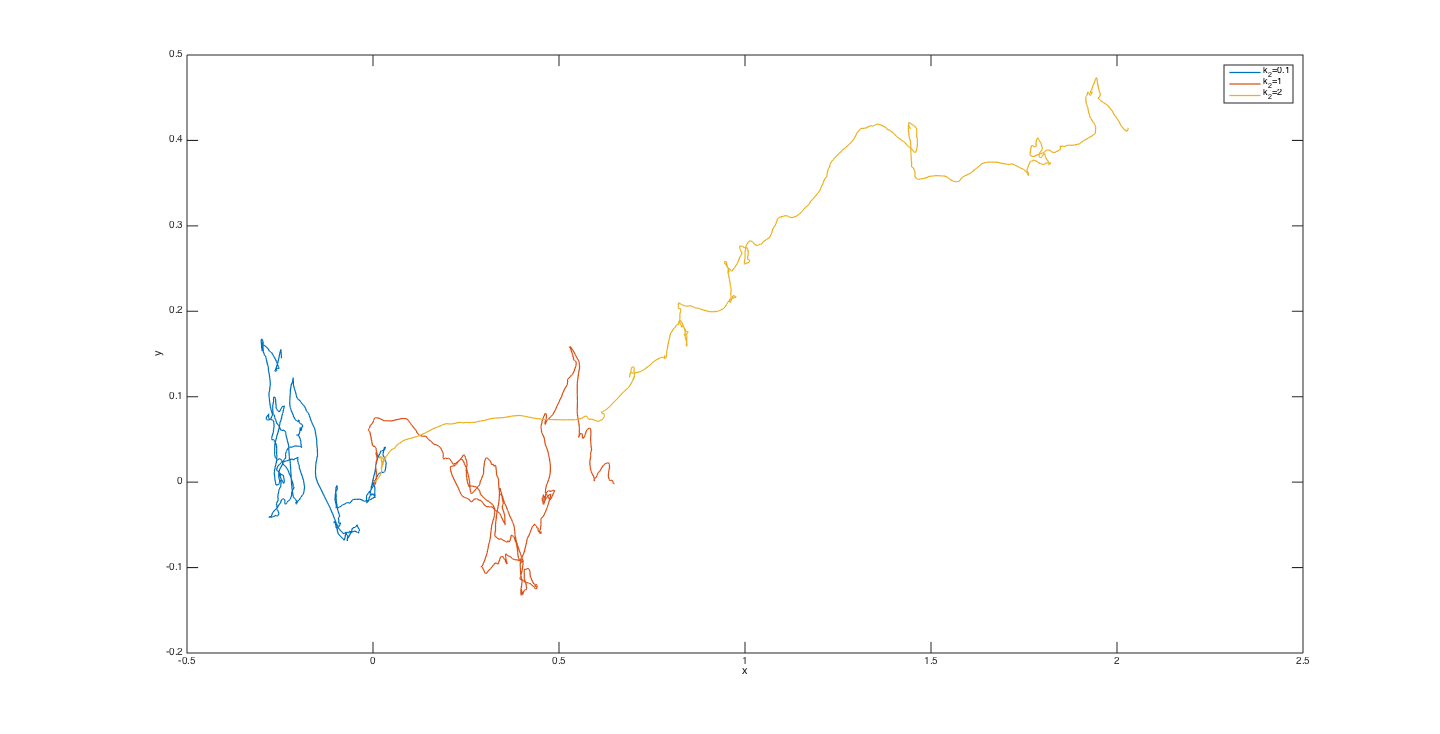}}\quad
	\subfloat[$\alpha=0.1$]{\includegraphics[scale=0.15]{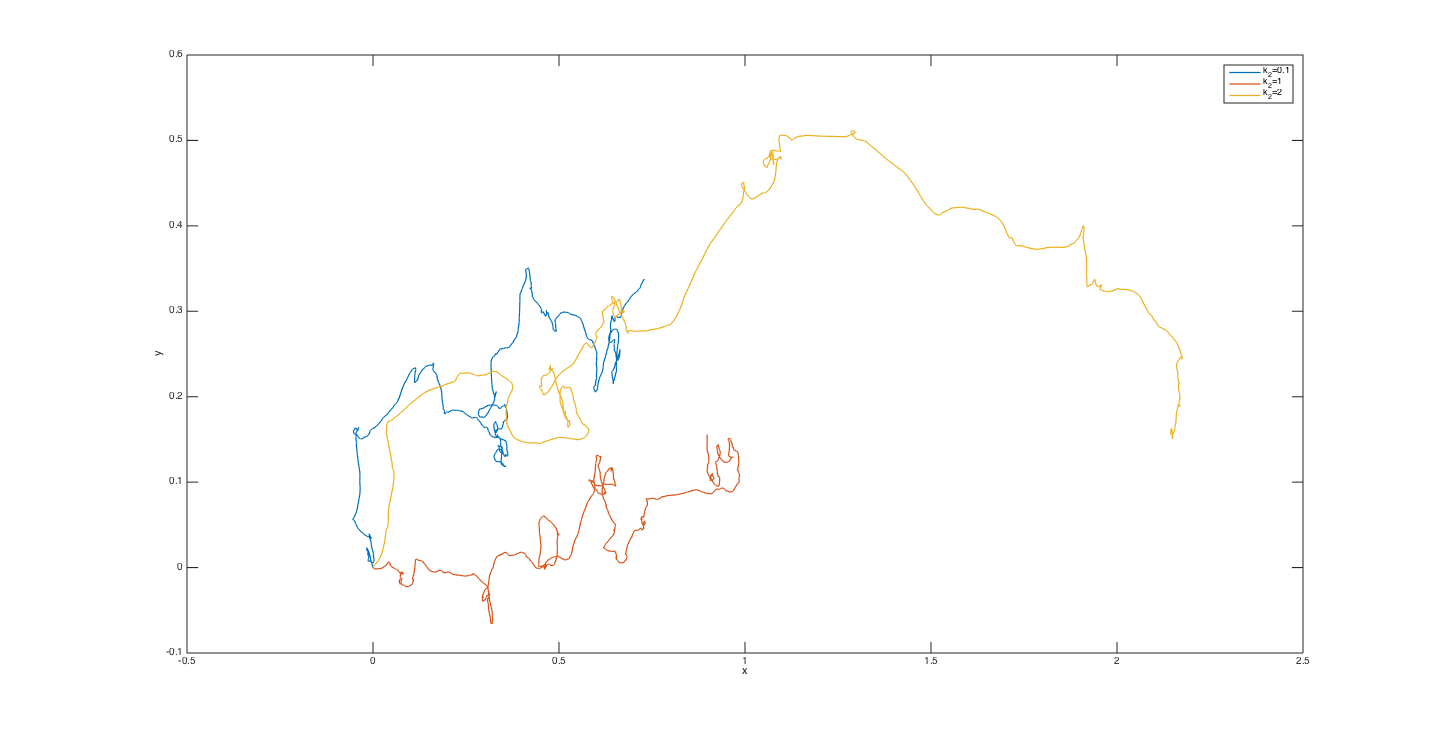}}\\	
		\subfloat[$\alpha=1$]{\includegraphics[scale=0.15]{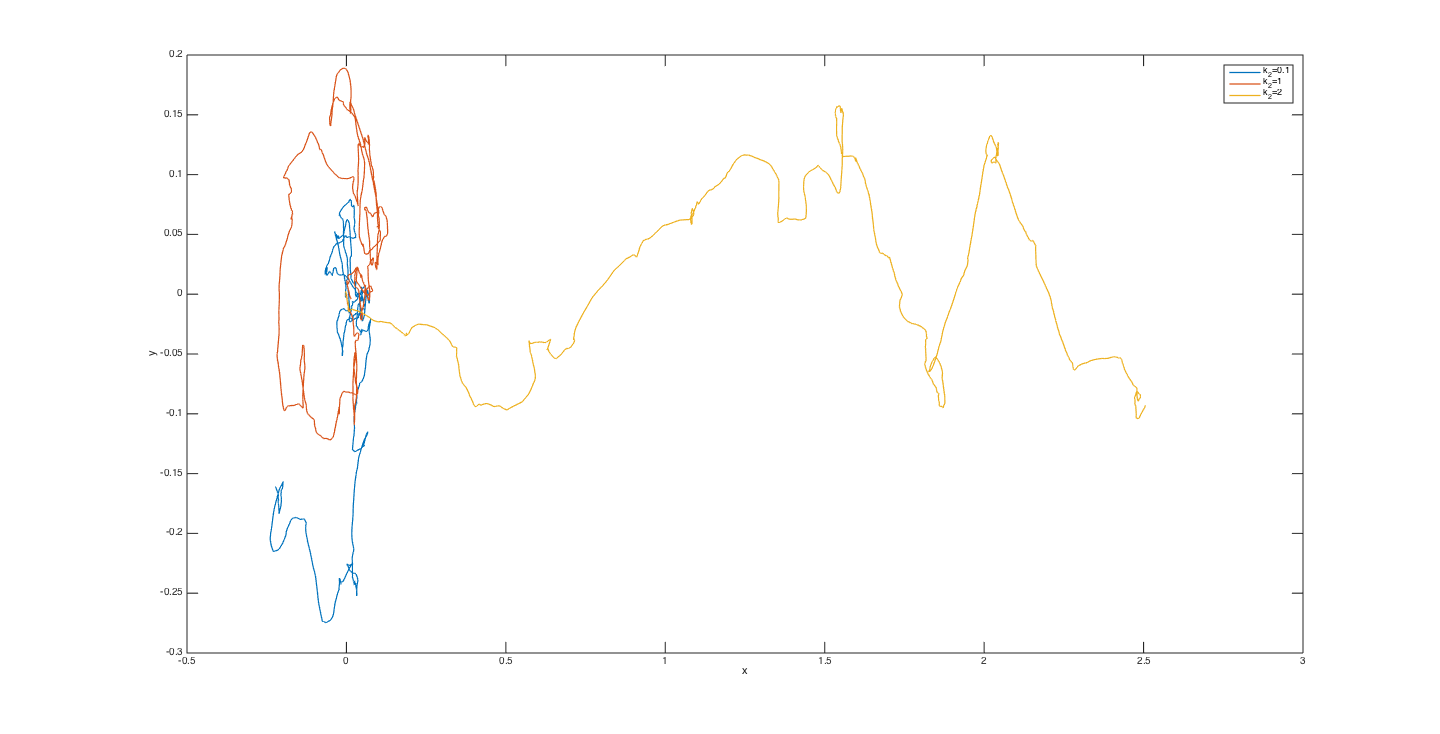}}\quad
	\subfloat[$\alpha=10$]{\includegraphics[scale=0.15]{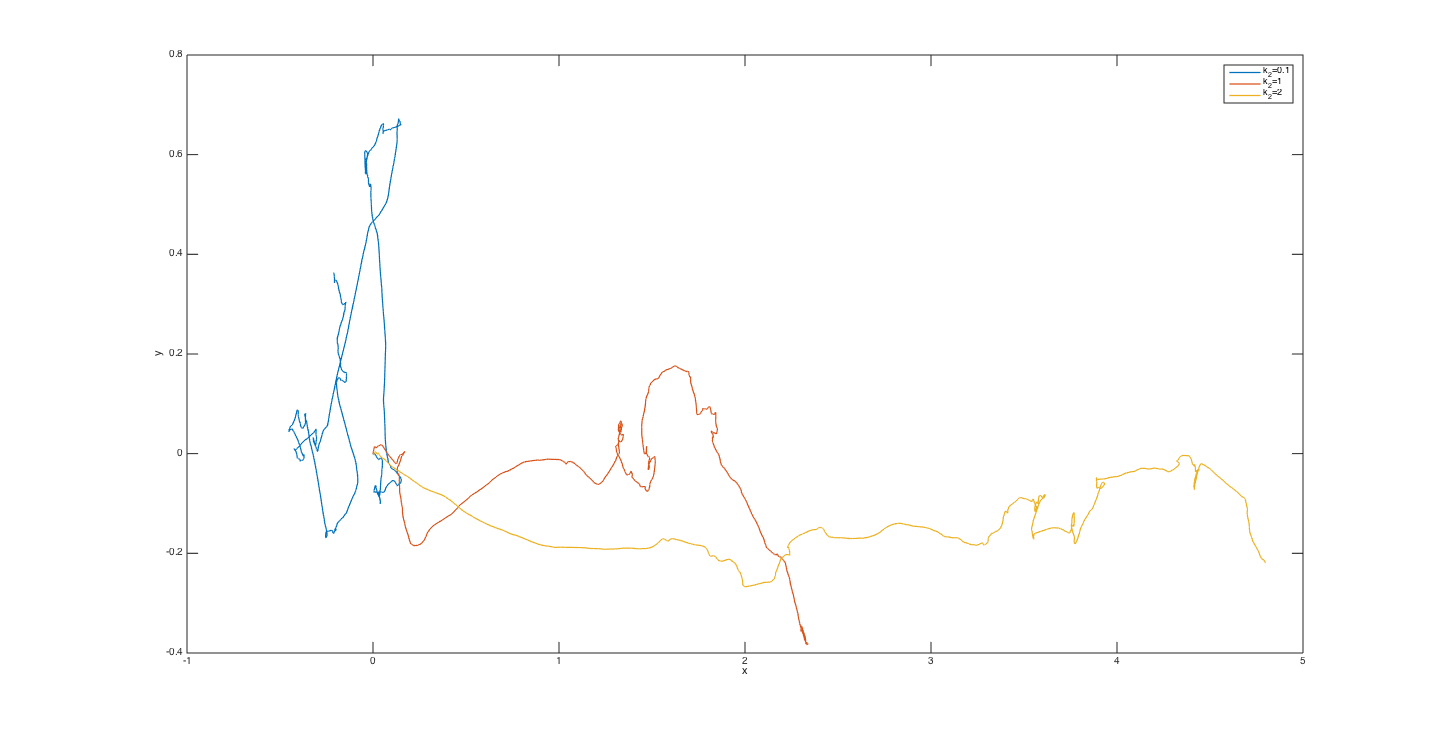}}\\	
	\caption{Numerical trajectories for an oscillating gradient of signal, a non polarised initial condition and several values for $\alpha$ and $\kappa_2$.}\label{fig:OnOffcos}
\end{figure}

The same type of behaviour as in the discontinuous rate can be observed. In particular, for $\alpha=0.1$, the succesive brownian and persistent phases are visible. Note that as explained before, the brownian phase for $\kappa_2=0.1$ corresponds to a persistent phase for $\kappa_2=2$, since each phenomenon prevents the other from winning.

\section{Conclusions and perspectives}
In this work, we have presented a 2D stochastic model for cell trajectories under chemical signalling. The model is based on an active particle model for cell trajectories developped in \citep{etchegaray:tel-01533458} that is able to provide a variety of trajectories without external signalling. 
Here, we take into account the positive guiding induced by a chemical gradient of signal favoring the protrusive activity in the same direction. For that purpose, we have adapted the reproduction rate in the protrusion population dynamics to deal with a bimodal potential. Moreover, we add a time dependence to the signal sensitivity so that the model captures the cell's reaction to time-dependent signals.
This work addresses the questions of the efficiency of guidance and of the balance with cells self-polarised internal machinery, that leads to non trivial behaviours. \par 
 
We have shown that the model is a well-posed non-homogeneous markovian process. Numerical simulations are performed using the thinning method to deal with the time-dependence. Finally, numerical experiments put to light some rich features of the model. Non intuitive features were observed. In particular, it is of interest that signalling can slow down a cell that is very poorly sensitive to it. Moreover, a cell can temporarily go in the wrong direction if it is highly polarised in the other direction.

Overall, this work shows that this model is able to capture non trivial cell behaviours. In further studies, other experimental settings will be investigated, such as gradients of signal switching directions, and space-dependent cues. In all cases, renormalization of the dynamics allows to derive a continuous model, for which theoretical information may be derived. The rigorous justification of this limit in the signalling case will be addressed. Finally, the question of optimal strategies to attain a target will be adressed, related to immune cells fate to reach pathogens in a body.

\newpage
\bibliographystyle{apalike}
\bibliography{Signal}
\end{document}